\documentclass{article}
\usepackage[a4paper,margin=1in]{geometry}
\usepackage{xcolor}
\usepackage[colorlinks]{hyperref}
\definecolor{darkred}{rgb}{0.5,0,0}
\definecolor{darkblue}{rgb}{0,0,0.5}
\definecolor{darkgreen}{rgb}{0,0.5,0}
\hypersetup{colorlinks
  ,linkcolor=darkred
  ,filecolor=darkgreen
  ,urlcolor=darkred
  ,citecolor=darkblue,
  linktocpage}
\usepackage{amsmath, amssymb, amsthm, mathtools,mathrsfs}
\usepackage{doi}
\usepackage{cleveref}
\usepackage{centernot}
\usepackage{tikz}
\usepackage[hyphenbreaks]{breakurl}

\newtheorem{theorem}{Theorem}[section]
\newtheorem{lemma}[theorem]{Lemma}
\newtheorem{corollary}[theorem]{Corollary}
\newtheorem*{claim}{Claim}

\newenvironment{claimproof}[1][Proof of the Claim.]{\proof[#1]}{\endproof}

\theoremstyle{definition}
\newtheorem{definition}[theorem]{Definition}
\newtheorem{example}[theorem]{Example}

\newcommand{\owlref}{\operatorname{owl-ref}}
\newcommand{\owl}{\operatorname{owl}}
\newcommand{\atp}{\operatorname{atp}}
\newcommand{\wl}{\operatorname{wl}}
\newcommand{\dom}{\operatorname{dom}}
\newcommand{\BP}{\operatorname{BP}}
\newcommand{\CR}{\operatorname{CR}}
\newcommand{\free}{\operatorname{free}}
\newcommand{\bound}{\operatorname{bound}}
\newcommand{\qr}{\operatorname{qr}}
\newcommand{\im}{\operatorname{im}}
\newcommand{\td}{\operatorname{td}}

\title{Finite Variable Counting Logics with Restricted Requantification\footnote{The research of the second and third author leading to these results has received funding from the European Research Council (ERC) under
the European Union’s Horizon 2020 research and innovation programme (EngageS: grant agreement No. 820148).}}

\newcommand{\mailadress}[1]{\href{mailto:#1@mathematik.tu-darmstadt.de}{\texttt{#1}}}

\author{Simon Raßmann \and Georg Schindling \\[4mm] TU Darmstadt \\
\(\mathclap{\{\mailadress{rassmann}, \mailadress{schindling}, \mailadress{schweitzer}\}\texttt{@mathematik.tu-darmstadt.de}}\)
\and Pascal Schweitzer}

\begin{document}

\maketitle

\begin{abstract}
    Counting logics with a bounded number of variables form one of the central concepts in descriptive complexity theory. Although they restrict the number of variables that a formula can contain, the variables can be nested within scopes of quantified occurrences of themselves. In other words, the variables can be requantified. We study the fragments obtained from counting logics by restricting requantification for some but not necessarily all the variables.

    Similar to the logics without limitation on requantification, we develop tools to investigate the restricted variants. Specifically, we introduce a bijective pebble game in which certain pebbles can only be placed once and for all, and a corresponding two-parametric family of Weisfeiler-Leman algorithms.
    We show close correspondences between the three concepts.

    By using a suitable cops-and-robber game and adaptations of the Cai-Fürer-Immerman construction, we completely clarify the relative expressive power of the new logics.

    We show that the restriction of requantification has beneficial algorithmic implications in terms of graph identification. Indeed, we argue that with regard to space complexity, non-requantifiable variables only incur an additive polynomial factor when testing for equivalence. In contrast, for all we know, requantifiable variables incur a multiplicative linear factor.

    Finally, we observe that graphs of bounded tree-depth and 3-connected planar graphs can be identified using no, respectively, only a very limited number of requantifiable variables.
\end{abstract}

\section{Introduction}

Descriptive complexity is a branch of finite model theory that essentially aims at characterizing how difficult logical expressions need to be in order to capture particular complexity classes. 
While we are yet to find or rule out a logic capturing the languages in the complexity class P,
there is an extensive body of work regarding the descriptive complexity of problems within~P. Most notably, there is the work of Cai, F\"{u}rer, and Immerman~\cite{cfi_opt} which studies a particular fragment of first-order logic.
This is the fragment~$\mathsf{C}^k$ in which counting quantifiers are introduced into the logic, but the number of variables is restricted to being at most~$k$.
The seminal result in~\cite{cfi_opt} shows that this logic fails to define certain graphs up to isomorphism,
which in turn proves that \emph{inflationary fixed-point logic with counting} \textsf{IFP+C} fails to capture~P.

Although the fragment~$\mathsf{C}^k$ restricts the number of variables, it is common for variables to be reused within a single logical formula. In particular, variables can be nested within scopes of quantified occurrences of themselves. In other words, they can be requantified. In our work, we are interested in understanding what happens if we limit the ability to reuse variables through requantification. In fact, we may think of reusability as a resource (in the vein of time, space, communication, proof length, advice etc.) that should be employed economically.

It turns out that the ability to limit requantification provides us with a more detailed lens into the landscape of descriptive complexities within P, much in the fashion of fine-grained complexity theory.

\subparagraph*{Results and techniques}
Let us denote by~$\mathsf{C}^{(k_1,k_2)}$ the fragment of first-order logic with counting quantifiers in which the formulas have at most~$k_1$ variables that may be requantified and at most~$k_2$ variables that may not be requantified.

First, we show that many of the traditional techniques of treating counting logics can be adapted to the setting of limited requantification. Specifically, it is well known that there is a close ternary correspondence between the logic~$\mathsf{C}^k$, the combinatorial bijective $k$-pebble game, and the famous $(k-1)$-dimensional Weisfeiler-Leman algorithm~\cite{cfi_opt,hella_logical_hierarchies, immerman_fo_approach}. We develop versions of the game and the algorithm that also have a limit on the reusability of resources. For the pebble game, a limit on requantification translates into pebbles that cannot be picked up anymore, once they have been placed. For the Weisfeiler-Leman algorithm, the limit on requantification translates into having some dimensions that ``cannot be reused''.
In fact the translation to the algorithmic viewpoint is not as straightforward as one might hope at first. Indeed, we do not know how to define a restricted version of the classical Weisfeiler-Leman algorithm that corresponds to the logic~$\mathsf{C}^{(k_1,k_2)}$. 
However, we circumvent this problem by employing the oblivious Weisfeiler-Leman algorithm (OWL).
This variant is often used in the context of machine learning. 
In fact, Grohe~\cite{grohe_gnn} recently showed that $k+1$-dimensional OWL is in fact exactly as powerful as $k$-dimensional (classical) WL.
We develop a resource-reuse restricted version of the oblivious algorithm and prove equivalence to our logic.
Indeed, we formally prove precisely matching correspondences between the limited requantification, limited pebble reusability, and the limited reusable dimensions (Theorem~\ref{thm:characterization}).

Next, we conclusively clarify the relation between the logics within the two-parametric family~$\mathsf{C}^{(k_1,k_2)}$. We show that in most cases limiting the requantifiability of a variable strictly reduces the power of the logic. We argue that no amount of requantification-restricted variables is sufficient to compensate the loss of an unrestricted variable. However, these statements are only true if at least some requantifiable variable remains. In fact, exceptionally, $\mathsf{C}^{(1,k_2)}$ is strictly less expressive than~$\mathsf{C}^{(0,k'_2)}$ whenever~$k_2'>2k_2$ (Theorem~\ref{thm:hierarchy}). To show the separation results, we adapt a construction of F\"{u}rer~\cite{fuerer_grid} and develop a cops-and-robber game similar to those in~\cite{Fluck2024, Grohe2023}. In this version, some of the cops may repeatedly change their location, while others can only choose a location once and for all. Using another graph construction, we rule out various a priori tempting ideas concerning normal forms in  $\mathsf{C}^{(k_1,k_2)}$. To this end we show that formulas in the logics can essentially become as complicated as possible, having to repeatedly requantify all of the requantifiable variables an unbounded number of times, before using a non-requantifiable variable (Corollary~\ref{cor:no:normalform}). In terms of the pebble game, it seems a priori unclear when an optimal strategy would employ the non-reusable pebbles. However, the corollary says that in general one has to conserve the non-reusable pebbles for possibly many moves until a favorable position calls for them.

Having gone through the technical challenges that come with the introduction of reusability, puts us into a position to discuss the implications.
Indeed, as our main result, we argue that our finer grained view on counting logics through restricted requantification has beneficial algorithmic implications. Specifically, we show that 
equivalence with respect to the logic $\mathsf{C}^{(k_1,k_2)}$ can be decided in polynomial time with a space complexity of \(O\bigl(n^{k_1}\log n\bigr)\), hiding
quadratic factors depending only on $k_1$ and $k_2$ (\Cref{thm:space:complexity}). This shows that while the requantifiable variables each incur a multiplicative linear factor in required space, the restricted variables only incur an additive polynomial factor. In particular,
equivalence with respect to the logics $\mathsf{C}^{(0,k_2)}$ can be decided in logarithmic space.
To show these statements, we leverage the fact that, because non-requantifiable variables cannot simultaneously occur free and bound,
the \(\mathsf{C}^{(k_1,k_2)}\)-type of a variable assignment does not depend on the \(\mathsf{C}^{(k_1,k_2)}\)-type of assignments which disagree regarding non-requantifiable variables.
Moreover, we use ideas from an algorithm of Lindell, which computes isomorphism of trees in logarithmic space~\cite{lindell_logspace} to implement the iteration steps of our algorithm. Generally, we believe the new viewpoint may be of interest in particular for applications in machine learning, where the WL-hierarchy appears to be too coarse for actual applications with graph neural networks (see for example~\cite{DBLP:conf/nips/BarceloGRR21,DBLP:conf/iclr/BevilacquaFLSCB22,DBLP:conf/nips/0001RM20,DBLP:conf/iclr/0002CWLF23}).
In the process of the space complexity proof, we also show that the iteration number of the resource-restricted Weisfeiler-Leman algorithm described above is at most \((k_2+1)n^{k_1}-1\) (\Cref{cor:iteration:number}).

Justifying the new concepts of restricted reusability, we observe that there are interesting graph classes that are identified by the logics~$\mathsf{C}^{(k_1,k_2)}$. We argue that \(\mathsf{C}^{(0,d+1)}\) identifies all graphs of tree-depth at most \(d\) (Theorem~\ref{thm:logic:vs:tree:depth}) and that \(\mathsf{C}^{(2,2)}\) identifies all \(3\)-connected planar graphs~(\Cref{thm:3-con:plan:graph:vs:logic}).

\subparagraph*{Outline of the paper} After briefly providing necessary preliminaries (Section~\ref{sec:preliminaries}) we formally introduce the logics $\mathsf{C}^{(k_1,k_2)}$, the pebble game with non-reusable pebbles, the \((k_1,k_2)\)-dimensional oblivious Weisfeiler-Leman algorithm, and prove the correspondence theorem between them (Section~\ref{sec:finite:var:log:res}). We then relate the power of the logics to each other and rule out certain normal forms (Section~\ref{sec:role:of:reuse}). 
We then analyze the space complexity (Section~\ref{sec:space:compl}) and finally provide two classes of graphs that are identified by our logics (\Cref{sec:planar:and:tree:depth}).

\subparagraph*{Further related work} In addition to the references above, let us mention related investigations. Over time, a large body of work on descriptive complexity has evolved. For insights into fundamental results regarding bounded variable logics, we refer to classic texts~\cite{DBLP:books/daglib/0095988,immerman_fo_approach,otto_bounded_var,DBLP:conf/asl/PikhurkoV09}. 
However, highlighting the importance of the counting logics~$\mathsf{C}^{k}$, let us at least mention the Immerman-Vardi theorem~\cite{DBLP:journals/iandc/Immerman86,DBLP:conf/stoc/Vardi82}. It says that 
on ordered structures, \emph{least fixed-point logic} \textsf{LFP} captures~P. Since \textsf{LFP} has the same expressive power as \textsf{IFP+C} on ordered structures, also \textsf{IFP+C}, whose expressive power is
closely related to the expressive power of the logics \(\mathsf{C}^k\), captures~P. We should also mention the work of Hella~\cite{hella_logical_hierarchies} introducing the bijective $k$-pebble game which forms the basis for our resource restricted versions. 

\emph{(Counting logics on graph classes)}
Because of the close correspondence between the logic $\mathsf{C}^k$ and the $(k-1)$-dimensional Weisfeiler-Leman algorithm, our investigations are closely related to the notion of the \emph{Weisfeiler-Leman dimension} of a graph defined in \cite{grohe_minors_book}. 
Given a graph $G$ this is the least number of variables $k$ such that $\mathsf{C}^{k+1}$ identifies $G$.
In particular, on every graph class of bounded Weisfeiler-Leman dimension, the corresponding finite variable counting logic captures isomorphism.
Graph classes with bounded Weisfeiler-Leman dimension include graphs with a forbidden minor~\cite{grohe_minors} and graphs of bounded rank-width (or equivalently clique width)~\cite{grohe_rank_width}, which in both cases is also shown to imply that \textsf{IFP+C} captures~P on these classes. 
For a comprehensive survey we refer to~\cite{Kiefer_2020}. Our observations for planar graphs follow from techniques bounding the number of variables required for the identification of planar graphs~\cite{kiefer_planar}. Other recent classes not already captured by the results on excluded minors and rank-width include, for example, some polyhedral graphs~\cite{Li2023WeisfeilerLeman}, some strongly regular graphs~\cite{Cai2023StronglyRegularGraphs}, and permutation graphs~\cite{Guo2023WeisfeilerLemanPerm}. 

\emph{(Logic and tree decompositions)} In \cite{Dvorak2010} and independently \cite{Dell2018} it was shown that $\mathsf{C}^k$-equivalence is characterized by homomorphism counts from graphs of tree-width at most $k-1$.
Likewise, homomorphism counts from bounded tree-depth graphs characterize equivalence in counting logic of bounded quantifier-rank~\cite{grohe_tree-depth}. 
Recently, these results were unified to characterize logical equivalence in finite variable counting logics with bounded quantifier-rank in terms of homomorphism counts~\cite{Fluck2024}.

\emph{(Space complexity)} Ideas underlying Lindell's logspace algorithm for tree isomorphism have also been used in the context of planar graphs~\cite{datta_logspace} and more generally bounded genus graphs~\cite{DBLP:conf/stoc/ElberfeldK14}. Similar results exist for classes of bounded tree-width~\cite{das_logspace,elberfeld_logspace}.

\emph{(Further recent results)} Let us mention some quite recent results in the vicinity of our work that cannot be found in the surveys mentioned above. 
Regarding the quantifier-rank within counting logics, there is a recent superlinear lower bound \cite{Grohe2023} improving Fürer's linear lower bound construction \cite{fuerer_grid}.  Further, very recent work on logics with counting includes results on rooted unranked trees~\cite{DBLP:conf/lics/HellingsGBG23} and inapproximability of questions on unique games~\cite{DBLP:conf/lics/Tucker-Foltz21}. Finally, there has been a surge in research on descriptive complexity within the context of machine learning (see~\cite{grohe_gnn,DBLP:conf/lics/Bergerem19,DBLP:phd/dnb/Bergerem23}).

\section{Preliminaries}\label{sec:preliminaries}

\paragraph*{General notation.}
For \(n\in\mathbb{N}_+\) we use \([n]\) to denote the \(n\)-element set \(\{1,\dots,n\}\). We use the notation \(\{\!\!\{v_1,\dotsc,v_n\}\!\!\}\) for \emph{multisets}.
For \(k_1,k_2 \in \mathbb{N}_+\), we fix the variable sets
\([x_{k_1}] \coloneqq \{x_1,\dotsc,x_{k_1}\}\),
\([y_{k_2}] \coloneqq \{y_1,\dotsc,y_{k_2}\}\), and
\([x_{k_1}, y_{k_2}] \coloneqq \{x_1,\dotsc,x_{k_1},y_1,\dotsc,y_{k_2}\}\).
Given a set \(V\), a \emph{partial function} \(\alpha \colon [x_{k_1}, y_{k_2}] \rightharpoonup V\) assigns to every variable \(z \in [x_{k_1}, y_{k_2}]\) at most one element \(\alpha(z) \in V\).
If \(\alpha\) does not assign an element to \(z\), we write \(\alpha(z) = \bot\). Also, we write \(\im(\alpha)\) for the \emph{image} of $\alpha$.
With a finite set \(V\) and \(\alpha\colon[x_{k_1},y_{k_2}] \rightharpoonup V\) we associate the total function \(\overline{\alpha}\colon[x_{k_1},y_{k_2}]\to V\mathbin{\dot\cup}\{\bot\}\), which we also view as a \([x_{k_1},y_{k_2}]\)-indexed \((k_1+k_2)\)-tuple.
For \(z \in [x_{k_1}, y_{k_2}]\) and \(v \in V\), the function \(\alpha[z/v]\) is defined as \(\alpha\) but with \(\alpha(z)\) replaced by \(v\).

\paragraph*{Graphs.}
A graph is a pair \(G=(V(G),E(G))\) consisting of a finite set \(V(G)\) of \emph{vertices} and a set \(E(G)\subseteq\binom{V(G)}{2}\) of \emph{edges}.
We write $|G|$ for the number of vertices, called the \emph{order} of $G$.
For a vertex $v \in V(G)$ we define the \emph{neighborhood} $N_G(v) \coloneqq \{w \in V(G) : \{v,w\} \in E(G)\}$ and the \emph{degree} $d_G(v) \coloneqq |N_G(v)|$ of $v$ in $G$.
We call $v$ \emph{universal} in $G$ if $N_G(v) = V(G) \setminus \{v\}$.

A \emph{colored graph} consists of a graph \(G\) and a coloring function \(\chi\colon V(G)\to C\) with a finite, ordered set \(C\) of \emph{colors}.
For a colored graph $G$ and vertices $v_1,\dotsc,v_n \in V(G)$ the graph $G_{(v_1,\dotsc,v_n)}$ is obtained by assigning new and distinct colors to the vertices $v_1,\dotsc,v_n$ in $G$.
The vertices $v_1,\dotsc,v_n$ are then called \emph{individualized}.

An isomorphism of (colored) graphs \(G\) and \(H\) is a bijection \(\varphi\colon V(G)\to V(H)\) that preserves edges, non-edges, and vertex-colors.

We denote the complete graph on $n$ vertices by $K_n$, that is, the graph with vertex set $[n]$ and all possible edges included.
A star of degree $n$ is a graph consisting of one universal vertex of degree $n$ and its neighbors of degree $1$.

\paragraph*{First-order logic with counting.}
First-order logic with counting \(\mathsf{C}\) is an extension of first-order logic by counting quantifiers \(\exists^{\geq k}\) for all \(k\in\mathbb{N}\). These intuitively state that there exist at least \(k\) distinct vertices satisfying the formula that follows.

Over the language of colored graphs with variable set \(\mathcal{V}\), formulas are inductively built up using:
\begin{itemize}
    \item atomic formulas, namely
          \begin{itemize}
              \item the formulas \(x=y\) for variables \(x,y\in\mathcal{V}\),
              \item the formula \(E(x,y)\) for variables \(x,y\in\mathcal{V}\), stating that the vertices \(x\) and \(y\) are adjacent,
              \item for every color \(c\) and variable \(x\in\mathcal{V}\), the formula \(U_c(x)\), stating that the vertex \(x\) has color \(c\),
          \end{itemize}
    \item negation (\(\neg\)), conjunction (\(\land\)), disjunction (\(\lor\)), and implication (\(\rightarrow\)),
    \item quantification over vertices via \(\forall\), \(\exists\) and the counting quantifiers \(\exists^{\geq k}\).
\end{itemize}
For a colored graph \(G\), a variable assignment \(\alpha\colon \mathcal{V}\to V(G)\), and a formula \(\varphi\in\mathsf{C}\),
we write \(G,\alpha\models\varphi\) if the graph \(G\) together with the variable assignment \(\alpha\) \emph{satisfies} the formula \(\varphi\).

The \emph{quantifier-rank} \(\qr(\varphi)\) of a formula \(\varphi\) is the maximum depth of nested quantifiers in the formula.
The set of \emph{free variables} $\free(\varphi)$ of a formula \(\varphi\) is defined as the set of all variables that occur outside the scope of a corresponding quantifier in $\varphi$.
If $\free(\varphi) = \emptyset$ the formula $\varphi$ is called a \emph{sentence}.
The set of \emph{bound variables} $\bound(\varphi)$ is the set of all variables that occur quantified in $\varphi$.

If we restrict to a finite set of $k \in \mathbb{N}_+$ variables, the resulting logic is called \emph{$k$-variable counting logic} and denoted by \(\mathsf{C}^k\).
For $r \in \mathbb{N}$ the \emph{quantifier-rank-$r$ counting logic} $\mathsf{C}_r$ is obtained by restricting formulas in $\mathsf{C}$ to quantifier-rank at most $r$.
The \emph{$k$-variable quantifier-rank-$r$ counting logic} is defined as $\mathsf{C}^k_r \coloneqq \mathsf{C}^k \cap \mathsf{C}_r$.
Note that in this logic, variables are allowed to be \emph{requantified}. For example, the following is a formula in \(\mathsf{C}^2_3\), stating
that every red vertex has a blue vertex at distance \(2\).
\[\forall x_1 \biggl(U_{\text{red}}(x_1)\rightarrow
	\exists x_2 \Bigl(E(x_1,x_2)\land
		\exists x_1 \bigl(E(x_2,x_1) \land U_{\text{blue}}(x_1)\bigr)
	\Bigr)
  \biggr).\]
As a general reference on finite variable logics, we refer to \cite{otto_bounded_var}.

\paragraph*{The Weisfeiler-Leman algorithm.}
Let \(k \geq 1\) and \(G\) be a colored graph.
The \emph{\(k\)-dimensional Weisfeiler-Leman algorithm} (short \(k\)-WL) iteratively computes a coloring of the \(k\)-tuples of vertices of \(G\). We also view the \(k\)-tuples as total functions \(\overline{\alpha} \colon \{x_1,\dotsc,x_k\} \to V(G)\).
Initially, each tuple \(\overline{\alpha} \in V(G)^k\) is colored by $$\wl^{(0)}_{k}(G, \overline{\alpha}) \coloneqq \operatorname{atp}_{k}(G, \overline{\alpha}).$$
Here, \(\operatorname{atp}_{k}(G,\overline{\alpha})\) is the \emph{atomic type} of \(\overline{\alpha}\) in \(G\),
i.e., the set of all atomic formulas with variables in $\{x_1,\dotsc,x_k\}$ satisfied by the colored graph \(G[\operatorname{im}(\alpha)]\) together with the assignment $\alpha$.
For every \(r\in\mathbb{N}\), we then inductively set
\begin{align*}
    \wl_{k}^{(r+1)}(G, \overline{\alpha}) \coloneqq
    \Bigl(\wl_k^{(r)}(G,\overline{\alpha}); \left\{\!\!\left\{ \left(\wl_{k}^{(r)}(G,\overline{\alpha}[x_i/u])\right)_{i\in[k]}: u \in V(G)\right\}\!\!\right\}\Bigr)
\end{align*}
whenever \(k \geq 2\), but for the case \(k=1\) we set
\[\wl_{k}^{(r+1)}(G, \overline{\alpha}) \coloneqq \Bigl(\wl_{k}^{(r)}(G, \overline{\alpha}) ; \left\{\!\!\left\{\wl_{k}^{(r)}(G, u)  : u \in N_G(\overline{\alpha})\right\}\!\!\right\}\Bigr).\]
We write \(\wl_k^{(r)}(G)\) for the coloring of all $k$-tuples of vertices of $G$ assigning $\wl_{k}^{(r+1)}(G, \overline{\alpha})$ to $\overline{\alpha}$.
Since the definition of \(\wl_{k}^{(r+1)}(G, \overline{\alpha})\) includes the color \(\wl_{k}^{(r)}(G, \overline{\alpha})\) of the previous iteration, the coloring \(\wl_{k}^{(r+1)}(G)\) \emph{refines} the coloring \(\wl_{k}^{(r)}(G)\).
That is,
whenever \(\wl_{k}^{(r+1)}(G, \overline{\alpha_1}) = \wl_{k}^{(r+1)}(G, \overline{\alpha_2})\)
for \(\overline{\alpha_1}, \overline{\alpha_2} \in V(G)^k\),
then we also have \(\wl_{k}^{(r)}(G, \overline{\alpha_1}) = \wl_{k}^{(r)}(G, \overline{\alpha_2})\).
Since there are exactly \(|V(G)|^k\)-many \(k\)-tuples of vertices, there exists an \(r \leq |V(G)|^k\) such that \(\wl_{k}^{(r)}(G)\) induces the same partition of color classes as \(\wl_{k}^{(r+1)}(G)\). It follows from the definition of the refinement
that this implies \(\wl_k^{(r)}(G)\) induces the same partition as \(\wl_k^{(r')}(G)\) for all \(r'\geq r\).
In this case we say that \(k\)-WL \emph{stabilizes} after at most \(r\) iterations on the graph \(G\).
If \(r\) is minimal with this property, we write \(\wl_{k}^{(\infty)}(G)\coloneqq\wl_k^{(r+1)}(G)\)
and call \(\wl_k^{(\infty)}(G)\) the \emph{stable coloring}.
For a second colored graph \(H\), we say that \(k\)-WL \emph{distinguishes} \(G\) and \(H\) after \(r\) iterations if there exists a color \(c\) such that
\[\left|\left\{\overline{\alpha} \in V(G)^k: \wl_k^{(r)}(G, \overline{\alpha}) = c\right\}\right| \neq
    \left|\left\{\overline{\beta}  \in V(H)^k: \wl_k^{(r)}(H, \overline{\beta})  = c\right\}\right|.\]

Besides the classical \(k\)-dimensional Weisfeiler-Leman algorithm, there also exists
a variant, called the \emph{\((k+1)\)-dimensional oblivious Weisfeiler-Leman algorithm} (short \((k+1)\)-OWL).
This variant colors \((k+1)\)-tuples of vertices, which we again view as total functions \(\overline{\alpha}\colon\{x_1,\dots,x_{k+1}\}\to V(G)\).
Initially, all tuples are colored by their atomic type:
\[\owl_{k+1}^{(0)}(G,\overline{\alpha})\coloneqq\atp_{k+1}(G,\overline{\alpha}).\]
Then, this coloring is iteratively refined by setting
\begin{align*}
    \owl_{k+1}^{(r+1)}(G,\overline{\alpha})\coloneqq\Bigl(\owl_{k+1}^{(r)}(G,\overline{\alpha}); \left\{\!\!\left\{\owl_{k+1}^{(r)}(G,\overline{\alpha}[x_i/u]): u\in V(G)\right\}\!\!\right\}_{i\in[k+1]}\Bigr).
\end{align*}
The stable coloring \(\owl_{k+1}^{(\infty)}(G)\) and the notion of distinguishing graphs is defined as for \(k\)-WL.

It turns out that \(k\)-WL and \((k+1)\)-OWL have the same distinguishing power:
\begin{lemma}[{\cite[Lemma A.1, Corollary V.7]{grohe_gnn}}] \label{lem:wl:vs:owl}
    Let \(G\) and \(H\) be graphs, \(\overline{\alpha}\in V(G)^{k+1}\) and \(\overline{\beta}\in V(H)^{k+1}\). Then the following are equivalent for every \(r\in\mathbb{N}\):
    \begin{enumerate}
        \item \(\owl_{k+1}^{(r)}(G,\overline{\alpha})=\owl_{k+1}^{(r)}(H,\overline{\beta})\),
        \item \(\atp_{k+1}(G,\overline{\alpha})=\atp_{k+1}(H,\overline{\beta})\) and for all \(i\in[k+1]\), we have
              \(\wl_k^{(r)}(G,\overline{\alpha}_{\neq i})=\wl_k^{(r)}(H,\overline{\beta}_{\neq i})\),
              where \(\overline{\alpha}_{\neq i}\) is the \(k\)-tuple obtained from \(\alpha\) by deleting
              the \(i\)-th entry.
    \end{enumerate}
    Moreover, two graphs are distinguished by \(k\)-WL if and only if they are distinguished by \((k+1)\)-OWL.
\end{lemma}

\paragraph*{The CFI construction.}
The CFI construction, introduced by Cai, Fürer and Immerman in their seminal paper \cite{cfi_opt},
is a method of constructing pairs of graphs that are hard for the Weisfeiler-Leman algorithm to distinguish.
The construction starts from a so-called \emph{base graph}, that is,
a connected and colored graph such that every vertex receives a unique natural number as color.
By our convention, the coloring induces a linear ordering on the vertices of the base graph.
The vertices and edges of the base graph are called \emph{base vertices} and \emph{base edges}, respectively.

From a base graph \(G\) the \emph{CFI graph} \(X(G)\) is constructed as follows:
Every base vertex \(v \in V(G)\) is replaced by a CFI gadget \(F(v)\) of degree \(d \coloneqq d_G(v)\), which consists of vertices \((v, \overline{a})\) for all \(\overline{a} \in \{0,1\}^d\) for which \(|\{i \in [d] : a_i = 1\}|\) is even. The gadget has no edges\footnote{There are generally two styles of defining the construction, we use the one from~\cite{fuerer_grid}, which has no edges in the gadgets.}.
Let \(e = \{u,v\}\) be a base edge and \(u,v\) have degree \(d\) and \(d'\) in \(G\) respectively.
When \(e\) is the \(i\)-th edge incident to \(u\) and the \(j\)-th edge incident to \(v\) according to the ordering of \(V(G)\), we add the edges \(\{ \{(u, \overline{a}), (v, \overline{b})\} : a_i = b_j\}\) to \(E(X(G))\).
The vertices of each gadget \(F(v)\) are colored with the color of the base vertex \(v\) in \(G\).

To \emph{twist} a base edge \(\{u,v\} \in E\) in \(X(G)\) means to replace every edge between the gadgets \(F(u)\) and \(F(v)\) by a non-edge and every non-edge by an edge.
For a set of base edges \(S \subseteq E\), the \emph{twisted CFI graph} \(\widetilde{X}_S(G)\) is obtained by twisting every edge \(e \in S\) in \(X(G)\).
We also write \(\widetilde{X}_e(G)\) for \(\widetilde{X}_{\{e\}}(G)\). Let us stress that the base graph~$G$ is connected.
\begin{lemma}[{\cite[Lemma 6.2]{cfi_opt}}] \label{lem:twist:iso}
    For all base edges \(e, e' \in E\) there exists an isomorphism \(\varphi_{e, e'} \colon \widetilde{X}_e(G) \to \widetilde{X}_{e'}(G)\).
    Moreover, for all sets of base edges $S, T \subseteq E$, the graphs $\widetilde{X}_T(G)$ and $\widetilde{X}_S(G)$ are isomorphic if and only if $|S| \equiv |T| \mod 2$.
\end{lemma}
Up to isomorphism, it is justified to speak of the CFI graph \(X(G) = \widetilde{X}_\emptyset(G)\) and the twisted CFI graph \(\widetilde{X}(G) \coloneqq \widetilde{X}_e(G)\) for any \(e \in E\).

We will analyze combinatorial games played on CFI graphs, for which we use the following notion of connectivity:
Two base edges \(e, e' \in E(G)\) are called \emph{connected} with respect to a partial function \(\gamma \colon [x_{k_1}, y_{k_2}] \rightharpoonup V(G)\)
if either \(e=e'\) or there exists a path \(v_0,v_1,\dotsc,v_{\ell} \in V(G)\) such that it holds \(e = \{v_0,v_1\}, e' = \{v_{\ell-1},v_{\ell}\}\) and \(\gamma(z) \notin \{v_1,\dotsc,v_{\ell-1}\}\) for all \(z \in [x_{k_1}, y_{k_2}]\).
A component of \(G, \alpha\) is a connected component of edges with respect to \(\alpha\).
For a partial function $\alpha \colon [x_{k_1}, y_{k_2}] \rightharpoonup V(X(G))$ we write $\alpha_1 \colon [x_{k_1}, y_{k_2}] \rightharpoonup V(G)$ for the function mapping $z$ to the base vertex of $\alpha(z)$.
The following is a consequence of \cite[Lemma 2.9, 2.13]{tuprints24244}:

\begin{lemma} \label{lem:twist:component}
    Let \(G\) be a base graph with base edges \(e,e' \in E(G)\) and let \(\alpha, \beta \colon [x_{k_1}, y_{k_2}] \rightharpoonup V(\widetilde{X}(G)) \) with \(\dom(\alpha) = \dom(\beta)\) and $\alpha_1 = \beta_1$.
    Then there exists an isomorphism \(\varphi_{e, e'} \colon \widetilde{X}_e(G) \to \widetilde{X}_{e'}(G)\) with $\varphi_{e, e'}(\alpha(z)) = \beta(z)$ for all $z \in [x_{k_1}, y_{k_2}]$
    if and only if the base edges $e$ and $e'$ are within the same connected component with respect to $\alpha_1 = \beta_1$.
    Furthermore, there exists a path $P$ from $e$ to $e'$ in $G$ such that $\varphi_{e,e'}$ is the identity on all gadgets $F(v)$ for $v \in V(G) \setminus V(P)$.
\end{lemma}

\section{Finite Variable Counting Logics with Restricted Requantification}\label{sec:finite:var:log:res}
When working in the logic $\mathsf{C}^k$, it is often necessary to \emph{requantify} variables in order to express certain properties.
We introduce finite variable first-order logic with counting quantifiers and \emph{restricted requantification} to study this issue.
We then define an Ehrenfeucht-Fraïssé-style game as an important tool for the analysis of the newly introduced logic by game-theoretic arguments.
Finally, we devise a variant of $k$-OWL that precisely captures the expressive power of the logic and game and prove a characterization that closely ties the reusable and non-reusable resources among these objects.
First, we give a precise definition of \emph{requantification}.

\begin{definition}
    Consider the counting logic $\mathsf{C}$ over a set of variables $\mathcal{V}$.
    A variable $x \in \mathcal{V}$ is said to be \emph{requantified} in a formula $\varphi \in \mathsf{C}$
    if either $x \in \free(\varphi) \cap \bound(\varphi)$ or if there exist a subformula $Q x \psi$ of $\varphi$
    and in turn a subformula $Q' x \chi$ of $\psi$ with $Q,Q' \in \{\forall, \exists\} \cup \{\exists^{\geq n} \colon n \in \mathbb{N}\}$.
    We define the logic $\mathsf{C}^{(k_1,k_2)}$ as the fragment of $\mathsf{C}$ over the fixed variable set $\mathcal{V} = [x_{k_1}, y_{k_2}]$
    consisting of those formulas in which the variables from $\{y_1,\dotsc,y_{k_2}\}$ are not requantified.
    The fragment of $\mathsf{C}^{(k_1,k_2)}$ with quantifier-rank at most $r \in \mathbb{N}$ is denoted by $\mathsf{C}^{(k_1,k_2)}_r$.
\end{definition}

\begin{example}
    Consider the following $\mathsf{C}^{(2,1)}_3$ formula:
    \begin{align*}
        \bigl(\exists y_1 \neg E(x_2,y_1)\bigr)\land
        \exists^{\geq 4} x_1 \Bigl(
        E(x_2,x_1) \land
        \exists y_1 \bigl(\neg E(x_1,y_1)\bigr) \land
        \forall x_2 \bigl(\neg E(x_2,x_1) \rightarrow \exists^{\geq 3} x_1 E(x_1,x_2)\bigr)
        \Bigr)
    \end{align*}
    expressing that the vertex $x_2$ is not universal and has at least four non-universal neighbors such that every non-neighbor of those has degree at least three.
    The variable $x_2$ is requantified in this formula since it occurs free and bound.
    The variable $x_1$ is requantified because the subformula $\exists^{\geq 3} x_1 E(x_1,x_2)$ occurs within the scope of the outermost quantification $\exists^{\geq 4} x_1$.
    The variable $y_1$ however is not requantified since neither of its quantifications occurs in the scope of the other.
\end{example}

The central question we will investigate in the following is how the non-requantifiability restriction affects the expressive power of the logic $\mathsf{C}^{(k_1,k_2)}$.
To this end, we use the notation \(\mathsf{C}^{(k_1,k_2)} \preceq \mathsf{C}^{(k'_1, k'_2)}\) if every pair of graphs distinguished by \(\mathsf{C}^{(k_1,k_2)}\) is also distinguished by \(\mathsf{C}^{(k'_1,k'_2)}\).
In the case of unrestricted requantification (i.e. $k_2=0$) it is clear that $\mathsf{C}^{(k_1,0)} \preceq \mathsf{C}^{(k_1+1,0)}$. In this terminology, the central result of \cite{cfi_opt} is that this relation is strict for all \(k_1\in\mathbb{N}\).
For the case of restricted requantification we make the simple observation that having more variables is at least as expressive as having fewer variables
(independent of their ability to be requantified). We also observe that requantifiable variables are at least as expressive as non-requantifiable variables.
That is, for all $k_1,k_2 \in \mathbb{N}$ with $k_1+k_2\geq1$ it holds that $$\mathsf{C}^{(k_1,k_2)} \preceq \mathsf{C}^{(k_1,k_2+1)} \preceq \mathsf{C}^{(k_1+1,k_2)}.$$

Also, observe that having only non-requantifiable variables (i.e. $k_1=0$) bounds the quantifier-rank to at most $k_2$ and in turn every sentence of quantifier-rank at most $k_2$ can be rewritten using at most $k_2$ non-requantifiable variables.
More precisely, we have $\mathsf{C}^{(0, k_2)} \equiv \mathsf{C}_{k_2}$. 

Next, we establish an Ehrenfeucht-Fraïssé-style game which closely corresponds to the power of the previously defined
logics with respect to distinguishing graphs.
The game is a variant of the bijective pebble game introduced in \cite{hella_logical_hierarchies} with the additional restriction that some pebbles may not be picked up again once placed.

\begin{definition}
    Suppose $k_1,k_2 \in \mathbb{N}$ and $k_1+k_2 \geq 1$.
    For colored graphs \(G\) and \(H\), we define the \emph{bijective $(k_1,k_2)$-pebble game} $\BP_{(k_1,k_2)}(G,H)$ as follows:

    The game is played by the players Spoiler, denoted by (S), and Duplicator, denoted by (D), with one pair of pebbles for each variable in $[x_{k_1},y_{k_2}]$.
    The pebble pairs in $[x_{k_1}]$ are called \emph{reusable} and the pebble pairs in $[y_{k_2}]$ are called \emph{non-reusable}.

    The game proceeds in rounds, each of which is associated with a pair of partial functions
    $$\alpha \colon [x_{k_1},y_{k_2}] \rightharpoonup V(G), \quad \beta \colon [x_{k_1},y_{k_2}] \rightharpoonup V(H)$$
    with $\dom(\alpha) = \dom(\beta)$. We call such a pair of partial functions a \emph{$(k_1,k_2)$-configuration} on the pair $G,H$.
    These functions indicate the placement of the pebble pairs on the graphs. For a pebble pair $z \in [x_{k_1},y_{k_2}]$ and vertices
    $v \in V(G), w \in V(H)$ we have $\alpha(z) = v, \beta(z)= w$ whenever the two pebbles of the pair $z$ are placed on \(v\) and \(w\), respectively.
    If not specified otherwise, both games start from the empty configuration given by $\dom(\alpha) = \dom(\beta) = \emptyset$.
    One round of the game with current configuration $(\alpha, \beta)$ consists of the following steps:
    \begin{enumerate}
        \item (S) picks up a pebble pair $z \in [x_{k_1},y_{k_2}]$ such that $z \in [x_{k_1}]$ or $\alpha(z)$ is undefined.
              If no such $z$ exists, the winning condition is checked directly.
        \item (D) chooses a bijection $f \colon V(G) \to V(H)$.
        \item (S) chooses $w \in V(G)$ and $f(w) \in V(H)$ to be pebbled with the pair $z$.
        \item The new configuration is given by \((\alpha[z / w], \beta[z / f(w)])\).
    \end{enumerate}
    The winning conditions are as follows:
    \begin{itemize}
        \item (S) wins immediately, if the initial configuration $(\alpha,\beta)$ does not induce a partial isomorphism.
              That is, the function $h \colon \im(\alpha) \to \im(\beta), \alpha(z) \mapsto \beta(z)$ is not a graph isomorphism from $G[\im(\alpha)]$ to $H[\im(\beta)]$.
        \item (S) wins if (D) cannot choose a bijection \(f\), i.e., if \(|G|\neq|H|\).
        \item (S) wins after the current round if the configuration $(\alpha,\beta)$ does not induce a partial isomorphism.
              Otherwise, the game continues and (D) wins the game if (S) never wins a round.
    \end{itemize}
    For $r \in \mathbb{N}_+$ we define the game variant $\BP_{(k_1,k_2)}^r$, which has the additional
    winning condition that (D) wins the game if (S) does not win after $r$ rounds.
    Furthermore, in the \(0\)-round game \(\BP_{(k_1,k_2)}^0\), player (D) wins if the initial configuration induces a partial isomorphism.
\end{definition}

We start our investigation with a simple observation on the distinguishing power of these games with few pebble pairs,
which nonetheless highlights the difference in expressive power of reusable and non-reusable pebbles.
The degree sequence of a graph $G$ is the multiset $\{\!\!\{d_G(v): v \in V(G)\}\!\!\}$ and the neighborhood-degree sequence is the multiset $\{\!\!\{ \{\!\!\{ d_G(w) : w \in N_G(v) \}\!\!\} : v \in V(G)\}\!\!\}$.
\begin{lemma} \label{lem:deg_seq}
    Let $G$ and $H$ be graphs, then (D) has a winning strategy for
    \begin{itemize}
        \item $\BP_{(0,2)}(G,H)$ if and only if the degree sequences of $G$ and $H$ coincide.
        \item $\BP_{(1,1)}(G,H)$ if and only if the neighborhood-degree sequences of $G$ and $H$ coincide.
    \end{itemize}
\end{lemma}
\begin{proof}
    We first note that in both games, if one non-reusable pebble pair is already placed, (D) has a winning strategy
    if and only if the two graphs have the same order and the two pebbled vertices have the same degree.
	Because of this, (D) has a winning strategy in the game \(\BP_{(0,2)}(G,H)\) if and only
	if there exists a degree-preserving bijection they can choose in the first round.	

	Next, consider the game \(\BP_{(1,1)}(G,H)\). Because it is never worth it for (S) to pick up the same
	pebble pair twice in a row, we can assume without loss of generality that they first pick up the reusable
	pebble pair, then the non-reusable pair and finally the reusable pair again.
	Thus, in order to win, (D) must ensure that their bijection in the second round is degree-preserving.
	But because this bijection must also preserve the neighborhood of the pebble pair placed in the first round,
	the bijection chosen by (D) in the first round must preserve the multisets of degrees of neighbors.
	Such a bijection exists if and only if the neighborhood-degree sequences of \(G\) and \(H\) coincide.
\end{proof}

We now turn to an algorithmic counterpart of the logic $\mathsf{C}^{(k_1,k_2)}$ and the game $\BP_{(k_1,k_2)}$.
It is an adaptation of the oblivious Weisfeiler-Leman algorithm \(k\)-OWL.

Indeed, to capture $\mathsf{C}^{(k_1,k_2)}$-equivalence, we iteratively color (partial) $(k_1+k_2)$-tuples of vertices of a given graph with the
previous color, and a sequence of multisets corresponding to variables as in $k_1$-OWL.
We deviate from the classical oblivious Weisfeiler-Leman algorithm by treating some entries of the tuple as \emph{non-reusable}:
For $y \in [y_{k_2}]$ with $\alpha(y) \neq \bot$, the variable $y$ is already assigned in the logic, respectively the non-reusable pebble is already placed in the game.
Thus, the entry in $\alpha$ corresponding to this variable should not be replaced by other vertices, but be kept fixed.
For this reason we utilize the advantage of oblivious Weisfeiler-Leman that each multiset corresponds to exactly one variable and pebble pair respectively.

Recall that we can view a variable assignment \(\alpha\colon[x_{k_1},y_{k_2}] \rightharpoonup V(G)\)
for a graph \(G\) as a \([x_{k_1},y_{k_2}]\)-indexed \((k_1+k_2)\)-tuple over \(V(G)\mathbin{\dot\cup}\{\bot\}\),
which we denote by \(\overline{\alpha}\).
For a variable assignment \(\alpha\), we set $J(\alpha) \coloneqq \{j \in [k_2]\colon \overline{\alpha}(y_j) = \bot\}$.
\begin{definition}
    Let $G$ be a graph and $k_1, k_2 \in \mathbb{N}$ with $k_1+k_2 \geq 1$.
    The \emph{\((k_1,k_2)\)-dimensional oblivious Weisfeiler-Leman algorithm} (short \((k_1,k_2)\)-OWL) iteratively computes a coloring
    of \([x_{k_1},y_{k_2}]\)-indexed \((k_1+k_2)\)-tuples over \(V(G)\mathbin{\dot\cup}\{\bot\}\).

    Initially, each tuple \(\overline{\alpha}\) is colored by its atomic type in \(G\):
    \[\owl^{(0)}_{(k_1,k_2)}(G, \overline{\alpha}) \coloneqq \operatorname{atp}_{k_1+k_2}(G, \overline{\alpha}).\]
    This coloring is then refined recursively: for every \(r\in\mathbb{N}\), we define
    \begin{align*}
        \owl^{(r+1)}_{(k_1,k_2)}(G,\overline{\alpha}) \coloneqq \Bigl(\owl^{(r)}_{(k_1,k_2)}(G, \overline{\alpha}),
         & \left\{\!\!\left\{ \owl^{(r)}_{(k_1,k_2)}(G, \overline{\alpha}[x_i/w])\colon w\in V(G) \right\}\!\!\right\}_{i\in [k_1]},            \\
         & \left\{\!\!\left\{ \owl^{(r)}_{(k_1,k_2)}(G, \overline{\alpha}[y_j/w]): w\in V(G) \right\}\!\!\right\}_{j\in J(\alpha)} \Bigr).
    \end{align*}
\end{definition}

Just as in the classical case, the coloring $\owl^{(r+1)}_{(k_1,k_2)}(G)$ refines $\owl^{(r)}_{(k_1,k_2)}(G)$ and eventually stabilizes.
We denote the stable coloring by $\owl^{(\infty)}_{(k_1,k_2)}(G)$.

The correspondence of counting logic, pebble game, and algorithm for restricted reusability now is as follows:

\begin{theorem} \label{thm:characterization}
    Let $G,H$ be colored graphs and $k_1,k_2 \in \mathbb{N}$ with $k_1 + k_2 \geq 1$.
    Then for all $(k_1,k_2)$-configurations $(\alpha, \beta)$ and $r \in \mathbb{N}$ the following are equivalent:
    \begin{enumerate}
        \item For every $\varphi \in \mathsf{C}^{(k_1,k_2)}_r$ with $\free(\varphi)\cap[x_{k_1}] \subseteq \dom(\alpha)\cap[x_{k_1}]$ and $\free(\varphi)\cap[y_{k_2}] = \dom(\alpha)\cap[y_{k_2}]$
              it holds that $G, \alpha \models \varphi \Leftrightarrow H, \beta \models \varphi$.
        \item (D) has a winning strategy for $\BP^r_{(k_1,k_2)}(G,H)$ with initial configuration $(\alpha, \beta)$.
        \item It holds that $\owl_{(k_1,k_2)}^{(r)}(G, \overline{\alpha}) = \owl_{(k_1,k_2)}^{(r)}(H, \overline{\beta})$.
    \end{enumerate}
\end{theorem}
\begin{proof}
	We proceed by induction on $r$.
	The base case $r=0$ holds by definition: $\owl_{(k_1,k_2)}^{(0)}(G, \overline{\alpha}) = \owl_{(k_1,k_2)}^{(0)}(H,\overline{\beta})$ holds exactly if $G,\alpha$ and $H,\beta$ satisfy the same quantifier-free formulas with variables in $\operatorname{dom}(\alpha)$.
	This is true if and only if $\alpha(z) \mapsto \beta(z)$ defines a partial isomorphism and in turn if (D) wins the zero move game with initial configuration $(\alpha, \beta)$.
	Now assume that the equivalence holds for all $(k_1,k_2)$-configurations $(\alpha, \beta)$ and for all $r<m$.

	$(1 \Rightarrow 3)$: Assume that $G, \alpha$ and $H, \beta$ agree on all formulas of $\mathsf{C}^{(k_1,k_2)}_m$.
	Then $G,\alpha$ and $H, \beta$ also agree on all formulas of $\mathsf{C}^{(k_1,k_2)}_{m-1}$ and by induction $\owl_{(k_1,k_2)}^{(m-1)}(G, \overline{\alpha}) = \owl_{(k_1,k_2)}^{(m-1)}(H, \overline{\beta})$.
	It remains to show that for every $z \in [x_{k_1}] \cup \{y_j: j \in J(\alpha)\}$ and corresponding colors $c_z$ it holds that
	\begin{equation} \label{eq:sizes}
		|\{v \in V(G): \owl^{(m-1)}_{(k_1,k_2)}(G,\overline{\alpha}[z / v]) = c_z\}| = |\{w \in V(H): \owl^{(m-1)}_{(k_1,k_2)}(H,\overline{\beta}[z / w]) = c_z\}|.
	\end{equation}
	By induction, the $k_1+k_2$ tuples \(\overline{\alpha}, \overline{\beta}\) are in the same $\owl_{(k_1,k_2)}^{(m-1)}$ color class exactly if they satisfy exactly the same $\mathsf{C}^{(k_1,k_2)}_{m-1}$ formulas
	with free variables in \(\dom(\alpha) = \dom(\beta)\).
	Over the class of graphs of order $|G|$ there are only finitely many inequivalent formulas in $\mathsf{C}^{(k_1,k_2)}_{m-1}$,
	so let $\Psi(c_z)$ be a finite set of such formulas characterizing the color class of $c_z$ and set $\psi_{c_z} \coloneqq \bigwedge_{\psi \in \Psi(c_z)} \psi$.
	By assumption, $G,\alpha$ and $H,\beta$ agree on the $\mathsf{C}^{(k_1,k_2)}_{m}$ formulas of the form $\exists^{=N_i} x_i \psi_{c_{x_i}}$ and $\exists^{=N_j} y_j \psi_{c_{y_j}}$
	where $\exists^{= N} z \varphi$ is an abbreviation for $\exists^{\geq N} z \varphi \land \neg \exists^{\geq N+1} z \varphi$.
	Then $N_i$ and $N_j$ are chosen as the sizes of the color classes of $c_{x_i}$ and $c_{y_j}$ respectively.
	Note that in $\psi_{c_{y_j}}$ the variable $y_j$ must occur free (and not bound) and hence $y_j$ is not requantified in the formula $\exists^{=N_j} y_j \psi_{c_{y_j}}$.

	$(3 \Rightarrow 2)$: Assume $\owl_{(k_1,k_2)}^{(m)}(G,\overline{\alpha}) = \owl_{(k_1,k_2)}^{(m)}(H,\overline{\beta})$.
	Then \Cref{eq:sizes} holds for all $z \in [x_{k_1}] \cup \{y_j: j \in J(\alpha)\}$ and every color $c_z$.
	We describe a winning strategy for (D):
	On the first move, (S) picks up some pebble pair, say $y_j$.
	Then (D) can pick a bijection $f \colon V(G) \to V(H)$ such that for every $v \in V(G)$ it holds
	\[ \owl^{(m-1)}_{(k_1,k_2)}(G,\overline{\alpha}[y_j / v]) = \owl^{(m-1)}_{(k_1,k_2)}(H,\overline{\beta}[y_j / f(v)]).\]
	Then (S) chooses some $v \in V(G)$ and the new configuration is given by $(\alpha[y_j / v], \beta[y_j / f(v)])$.
	Since the $\owl_{(k_1,k_2)}^{(m-1)}$ colors of the tuples $\overline{\alpha}[y_j / v], \overline{\beta}[y_j / f(v)]$ coincide by assumption, (D) has a winning strategy for the remaining $m-1$ moves by induction.

	$(\neg 1 \Rightarrow \neg 2)$: Assume $G, \alpha \models \varphi$ and $H, \beta \centernot\models \varphi$ holds for some $\varphi \in \mathsf{C}^{(k_1,k_2)}_m$.
	If $\varphi$ is a disjunction, $G,\alpha$ and $H, \beta$ differ on at least one of the disjuncts and if $\varphi = \neg \psi$ we have $G, \alpha \centernot\models \psi$ and $H, \beta \models \psi$, so assume $\varphi = \exists^{\geq N} z \psi$ with $\psi \in \mathsf{C}^{(k_1,k_2)}_{m-1}$.
	We provide a winning strategy for (S):
	If $\varphi = \exists^{\geq N} x_i \psi$, player (S) picks up the $x_i$ pebble pair on the first move.
	Then (D) responds with a bijection $f \colon V(G) \to V(H)$ and since
	$$|\{v \in V(G): G, \alpha[x_i / v] \models \psi\}| \geq N > |\{w \in V(H): H, \beta[x_i / w] \models \psi\}|$$
	there exists some vertex $v \in V(G)$ such that $G, \alpha[x_i / v] \models \psi$ and $H, \beta[x_i / f(v)] \centernot\models \psi$.
	Player (S) chooses the pebble pair $x_i$ to be placed on $v, f(v)$.
	The new configuration is given by $(\alpha[x_i / v], \beta[x_i / f(v)])$ and differs on the formula $\psi \in \mathsf{C}^{(k_1,k_2)}_{m-1}$.
	By induction, (S) has a winning strategy for the remaining game in $m-1$ moves from this position.
	If $\varphi = \exists^{\geq N} y_j \psi$, the pebble pair $y_j$ is not yet placed since $y_j \notin \operatorname{free}(\varphi) \cap [y_{k_2}] = \operatorname{dom}(\alpha) \cap [y_{k_2}]$.
	Thus, (S) picks up the pebble pair $y_j$ and (D) chooses a bijection $f \colon V(G) \to V(H)$.
	Again, there exists some $v \in V(G)$ such that $G, \alpha[y_j / v] \models \psi$ and $H, \beta[y_j / f(v)] \centernot\models \psi$.
	Thus, (S) chooses the pebble pair $y_j$ to be placed on $v, f(v)$ and wins the remaining $m-1$ move game by induction.
\end{proof}

\section{The Role of Reusability}\label{sec:role:of:reuse}

We investigate the interplay of requantifiable and non-requanti\-fiable variables in $\mathsf{C}^{(k_1,k_2)}$ using the game-theoretic characterization provided by \Cref{thm:characterization}.
To this end, we utilize the CFI construction and introduce a variant of the cops-and-robber game used in \cite{Grohe2023} to simulate the game \(\BP_{(k_1,k_2)}\) on CFI graphs via a game played only on the base graph.
Our variant involves non-reusable cops as a way of restricting reusability of resources.

\begin{definition}
    The \emph{cops-and-robber game} $\CR_{(k_1,k_2)}(G)$ is played on a base graph $G$ between a group of $k_1+k_2$ \emph{cops} and one \emph{robber}.
    The cops are denoted by the elements of $[x_{k_1}, y_{k_2}]$ and a cop $x_i$ is called \emph{reusable} while a cop $y_j$ is called \emph{non-reusable}.
    Each round of the game is associated with a partial function $\gamma \colon [x_{k_1},y_{k_2}] \rightharpoonup V(G)$ and an edge $e \in E(G)$.
    The function $\gamma$ encodes the current positions of the cops while the edge $e$ is the position of the robber.
    Initially, there are no cops on the vertices and the robber is placed on some edge of the base graph.
    One round of the game with current position $(\gamma, e)$ consists of the following steps:
    \begin{enumerate}
        \item The cops choose $z \in [x_{k_1}, y_{k_2}]$ such that $z \in [x_{k_1}]$ or $\gamma(z)$ is undefined.
        If no such $z$ exists, the winning condition is checked directly. 
        Then a destination $w \in V(G)$ for $z$ is declared.
        \item The robber chooses an edge $e'$ in the connected component of $e$ with respect to~$\gamma[z/\bot]$.
        \item The cop $z$ is placed on the vertex $w$.
        \item The new position of the game is given by $(\gamma[z/w], e')$.
    \end{enumerate}
    The winning condition is as follows:
    \begin{itemize}
        \item  The cops win the game if at the end of the current round both vertices incident to the robber edge $e'$ hold cops.
              The robber wins if the cops never win.
    \end{itemize}
    We also introduce the game $\CR^r_{(k_1, k_2)}(G)$ with the additional winning condition that the robber wins if the cops do not win in $r$ rounds.
\end{definition}

Using this game, plays of the bijective $(k_1,k_2)$-pebble game on CFI graphs $X(G), \widetilde{X}(G)$ can be simulated by plays of the cops-and-robber game played only on the base graph $G$.
Intuitively, Spoiler has to \emph{catch} the twist in $\widetilde{X}(G)$ with pebbles to show the difference of the graphs.
This corresponds to moving the cops (according to the reusability of the used pebbles) in $\CR_{(k_1,k_2)}(G)$.
Duplicator, however, moves the twist in $\widetilde{X}(G)$ using automorphisms of the graph to hide the difference, which corresponds to moving the robber in $\CR_{(k_1,k_2)}(G)$.
Following similar arguments from \cite{Dawar2007, fuerer_grid}, this yields the following lemma, stating that the bijective pebble game on CFI graphs can be simulated appropriately.
\begin{lemma}\label{lem:CR:vs:BP}
    Let $k_1+k_2 \geq 2$ and $r \in \mathbb{N}$. Then the robber has a winning strategy in $\CR^r_{(k_1,k_2)}(G)$ if and only if (D) has a winning strategy in $\BP^r_{(k_1, k_2)}(X(G), \widetilde{X}(G))$.
\end{lemma}
\begin{proof}
    For the first part, we turn a winning strategy of the robber in $\CR^r_{(k_1,k_2)}(G)$ into a winning strategy of (D) in $\BP^r_{(k_1, k_2)}(X(G), \widetilde{X}(G))$.
    By \Cref{lem:twist:iso} we might assume that $\widetilde{X}(G) = \widetilde{X}_{e_0}(G)$ for some edge $e_0 \in E(G)$ and the robber is initially placed on $e_0$ in $\CR^r_{(k_1,k_2)}(G)$.
    For a $(k_1,k_2)$-configuration $(\alpha, \beta)$ on $X(G), \widetilde{X}(G)$ we write $\alpha_1,\beta_1 \colon [x_{k_1}, y_{k_2}] \rightharpoonup V(G)$ for the mapping to the corresponding base vertices.
    For every vertex $v \in V(G)$ that is declared as destination for a cop $z$ in $\CR^r_{(k_1,k_2)}(G)$ with cop positions $\gamma = \alpha_1 = \beta_1$, the robber selects some edge $e^v \in E(G)$ to escape to.
    With each round $\ell \in [r]$ of $\BP^r_{(k_1, k_2)}(X(G), \widetilde{X}_{e_0}(G))$ we associate an edge $e_{\ell} \in E(G)$ and an isomorphism $f_{\ell} \colon  \widetilde{X}_{e_{\ell}}(G) \to \widetilde{X}_{e_0}(G)$.
    Since $|X(G)| = |\widetilde{X}(G)|$ the game $\BP^r_{(k_1, k_2)}(X(G), \widetilde{X}_{e_0}(G))$ starts off properly.
    In the first round, (S) picks up a pebble pair $z_1 \in [x_{k_1}, y_{k_2}]$.
    By \Cref{lem:twist:component}, for each destination $v \in V(G)$ of the cop $z_1$ there exists an isomorphism $\varphi_{e^v, e_0} \colon \widetilde{X}_{e^v}(G) \to \widetilde{X}_{e_0}(G)$ with $\varphi_{e_0,e^v} \circ \alpha = \beta$.
    Then (D) chooses the bijection $f \colon X(G) \to \widetilde{X}_{e_0}(G), (v, \overline{a}) \mapsto \varphi_{e^v, e_0}((v, \overline{a}))$.
    Next, (S) chooses a vertex $(v_1, \overline{a}) \in V(X(G))$ and the pebble pair $z_1$ is placed on $(v_1, \overline{a}), f((v_1, \overline{a}))$.
    We set $e_1 \coloneqq e^{v_1}$ and $f_1 \coloneqq \varphi_{e^{v_1}, e_0}$, which is an isomorphism between $\widetilde{X}_{e_1}(G)$ and $\widetilde{X}_{e_0}(G)$. 
    At the beginning of round \(1 < \ell \leq r\) with current $(k_1,k_2)$-configuration $(\alpha, \beta)$, let (S) pick up the pebble pair $z \in [x_{k_1},y_{k_2}]$.
    Note that the reusability restrictions for the pebble pairs and corresponding cops are identical. 
    Again, for each $v \in V(G)$ declared as destination for the cop $z$ let $\varphi_{e^v, e_{\ell-1}} \colon \widetilde{X}_{e^v} \to \widetilde{X}_{e_{\ell-1}}$ denote the isomorphism from \Cref{lem:twist:component}.
    Since $\varphi_{e^v, e_{\ell-1}}((v, \cdot))$ is a bijection from $F(v)$ to itself and $f_{\ell-1}$ is a gadget-preserving bijection by induction,
    choosing $f \colon (v, \overline{a}) \mapsto f_{\ell-1}(\varphi_{e^v, e_{\ell-1}}(v, \overline{a}))$ is a valid choice for (D) that again preserves gadgets.
    Then (S) selects a vertex $(v_{\ell}, \overline{a}) \in V(X(G))$ and the configuration is updated to $(\alpha', \beta') = (\alpha[z / (v_{\ell}, \overline{a})], \beta[z/f(v_{\ell}, \overline{a})])$.
    Since the robber is not caught at this point in $\CR^r_{(k_1,k_2)}(G)$, we have $e^{v_{\ell}} \not\subseteq \im(\alpha'_1) = \im(\beta'_1)$.
    We set $e_{\ell} \coloneqq e^{v_{\ell}}$ and $f_{\ell} \coloneqq f_{\ell-1} \circ \varphi_{e_{\ell}, e_{\ell-1}}$.
    Then it holds $f_{\ell}(v_{\ell}, \overline{a}) = f(v_{\ell}, \overline{a})$ and since no vertex on the path from $e_{\ell}$ to $e_{\ell -1}$ for $\varphi_{e_{\ell}, e_{\ell-1}}$ is pebbled (occupied by a cop), also $f_{\ell} \circ \alpha' = \beta'$.
    Since $f_{\ell}$ restricts to an isomorphism from $X(G)$ (which results by twisting $e_{\ell}$ in $\widetilde{X}_{e_{\ell}}(G)$) to $\widetilde{X}_{e_0}(G)$ with $e_{\ell}$ twisted additionally, the map $\alpha'(z) \mapsto \beta'(z)$ is a partial isomorphism. 
    The robber has a winning strategy for the remaining game $\CR^{r-\ell}_{(k_1,k_2)}(G)$ with starting position $\alpha'_1, e_{\ell}$ and hence
    (D) inductively has a winning strategy for $\BP^{r-\ell}_{(k_1, k_2)}(X(G), \widetilde{X}_{e_0}(G))$ starting from $(\alpha', \beta')$.

    For the second part, assume that he cops have a winning strategy in $\CR^r_{(k_1, k_2)}(G)$. We provide a winning strategy for (S) in $\BP^r_{(k_1, k_2)}(X(G), \widetilde{X}(G))$.
    First, we assume that (D) always chooses bijections $f \colon V(X(G)) \to V(\widetilde{X}(G))$ with the property that one can apply an additional twist to the pebbled graph $f((X(G), \alpha)) = (f(X(G)), f \circ \alpha)$ such that the resulting pebbled graph is isomorphic to $\widetilde{X}(G), \beta$.
    In a round $\ell \leq r$ of $\BP^r_{(k_1, k_2)}(X(G), \widetilde{X}(G))$, player (S) picks up the pebble pair $z$ corresponding to the cop being picked up in $\CR^r_{(k_1,k_2)}(G)$.
    For every bijection $f$ chosen by (D) there exists an edge $e'$ in $f(X(G)), f \circ \alpha$ such that twisting $e'$ results in a graph isomorphic to $\widetilde{X}(G), \beta$.
    Then (S) can simulate this move in the cops-and-robber game by placing the robber on the edge $e'$ that must be in the same component as $e$ by \Cref{lem:twist:component}.
    In $\CR_{(k_1,k_2)}(G)$ the cop $z$ is moved to some vertex $v \in V(G)$ according to the winning strategy of the cops.
    Consequently, (S) selects vertices $(v, \overline{a}), f((v, \overline{a}))$ to be pebbled with the pair $z$.
    If the cops win $\CR^{r}_{(k_1,k_2)}(G)$ at this point, both endpoints $v,v'$ of $e'$ are occupied by cops.
    Therefore, the gadgets $F(v), F(v')$ both hold pebbles but since twisting $e'$ in $(f(X(G)), f \circ \alpha)$ results in $\widetilde{X}(G), \beta$ these pebbles induce an edge and a non-edge in $X(G)$ and $\widetilde{X}(G)$ respectively.
    If the cops do not win at this point, there is a winning strategy for the cops in the remaining $r-\ell$ rounds.
    Now assume that no twist of a base edge of $f((X(G), \alpha))$ results in a graph isomorphic to $(\widetilde{X}(G), \beta)$.
    If there is a pair of pebbles $z$ such that $f(\alpha(z))$ and $\beta(z)$ are in different gadgets, (S) wins by choosing $\alpha(z), f(\alpha(z))$.
    Otherwise, some gadget $F(u)$ in $f(X(G))$ contains two pebbles $z_1,z_2$ with $f(\alpha(z_1)) = (u, \overline{a}), f(\alpha(z_2)) = (u, \overline{b})$ and $a'_i = b'_i$
    whereas for $\beta(z_1) = (u, \overline{a'}), \beta(z_2) = (u, \overline{b'})$ it holds $a'_i \neq b'_i$.
    Hence, $f(\alpha(z_1)), f(\alpha(z_2))$ have a common neighbor in a gadget $F(v)$ while $\beta(z_1),\beta(z_2)$ do not have a common neighbor in $F(v)$.
    Then (S) chooses the common neighbor in $X(G)$ to be pebbled and wins the game.
\end{proof}

We now prove a strict hierarchy for the logics \(\mathsf{C}^{(k_1,k_2)}\) by providing, for every pair of logics we want to separate, two CFI graphs $X(G), \widetilde{X}(G)$ that are distinguished by one of the logics, but not the other.
To show this, it now suffices to provide strategies for the game $\CR_{(k_1, k_2)}(G)$ by \Cref{lem:CR:vs:BP} and \Cref{thm:characterization}.
The idea for the choice of the base graphs $G$ is inspired by \cite{fuerer_grid} where grid graphs were chosen as base graphs to show a linear lower bound on the number of iterations required by $k$-WL to distinguish non-isomorphic graphs.

\begin{definition}
    The graph
    \[G_{h \times \ell} \coloneqq (\{v_{i,j}: i \in [h], j \in [\ell]\}, \{\{v_{i,j}, v_{r,s}\}: |i-r| + |j-s|=1\})\]
    is called the \textit{grid graph with $h$ rows and $\ell$ columns} (see \Cref{fig:grid}). We also call $h$ the \emph{height} and $\ell$ the \emph{length} of the grid.

    We say that the cops \emph{build a barrier} in the game $\CR_{(k_1, k_2)}$ played on a graph $G$ containing a grid if they are placed on a separator of the graph $G$
    disconnecting the first column from the last column of the grid.
\end{definition}

When the cops-and-robber game is played on the grid graph $G_{h \times \ell}$ with at least $h+1$ cops, the cops can be placed on a column to form a barrier. They can move through the grid maintaining this formation by using the additional cop.
However, moving the barrier by one decreases the size of the component containing the robber only by one. In particular, the component size decreases at most by a constant per round.
Considering the setting with non-reusable cops, it is clear that these cannot be used to move the barrier through the graphs for an arbitrarily large distance and hence adding them to the game might not increase the capability to distinguish non-isomorphic graphs of a grid-type.

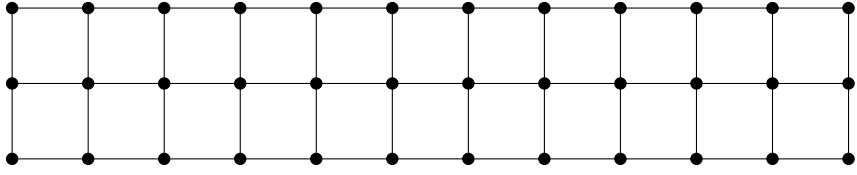
\begin{figure}
    \centering
    \begin{tikzpicture}
        \def\height{2} % Set height of the grid
        \def\length{11} % Set length of the grid

        % Draw nodes and edges
        \foreach \x in {0,...,\length} {
            \foreach \y in {0,...,\height} {
                % Draw a node at each grid point
                \node[circle, draw, fill=black, inner sep=1.5pt] (\x,\y) at (\x, \y) {};
                
                % Draw edges to the right and above nodes (horizontal and vertical edges)
                \ifnum\x<\length
                    \draw (\x,\y) -- ({\x+1},\y);
                \fi
                \ifnum\y<\height
                    \draw (\x,\y) -- (\x,{\y+1});
                \fi
            }
        }
    \end{tikzpicture}
    \caption{A drawing of the graph $G_{3 \times 12}$}
    \label{fig:grid}
\end{figure}

Using the established framework, we precisely classify how the reusability of variables affects the relative expressive power of the logic $\mathsf{C}^{(k_1,k_2)}$.
As noted before, increasing the total number of variables or exchanging non-requantifiable variables for requantifiable ones may only increase the expressive power.
It turns out that in general for most parameters combinations, these are the only possibilities to increase the expressive power. However, as we will see, there are certain special cases for which this is not true.

First, we show the advantages of reusability:
When the cops-and-robber game is played on $G_{h \times \ell}$ with at least $h+1$ reusable cops, the cops can be placed on a column to form a barrier. The cops can move through the grid maintaining this formation by using the additional cop.
To show the separation, we choose the base graph as a grid of sufficient length such that the cops are required to move or build a barrier repeatedly. This makes reusability necessary as all non-reusable cops are placed at some point and can not be used to move a barrier any further.

\begin{lemma} \label{lem:reusable:advantage}
    For all \(k_1,k_2,k'_1,k'_2\geq 0\), if \(k_1>\max(k_1',1)\), then $\mathsf{C}^{(k_1,k_2)} \not\preceq \mathsf{C}^{(k_1',k_2')}$.
\end{lemma}
\begin{proof}
    First, consider the game $\CR_{(k_1,k_2)}(G_{(k_1-1) \times (k_1 2^{2k'_2}+1)})$.
	The cops can build a barrier in the middle of the grid and move it towards the robber only using reusable cops.
	This is a winning strategy for the cops since the size of the component containing the robber is decreased by a constant in each round and eventually vanishes. 

    On the other hand, we show that the robber has a winning strategy for $\CR_{(k'_1,k'_2)}(G_{(k_1-1) \times (k_1 2^{2k'_2}+1)})$ by induction on $k'_2$.
    The base case for $k'_2=0$ is the game $\CR_{(k'_1,0)}(G_{(k_1-1) \times (k_1+1)})$, for which the robber has a winning strategy as a barrier can be built, but not moved. 
    For the inductive step assume $k'_2 > 0$ and consider the game $\CR_{(k'_1,k'_2+1)}(G_{(k_1-1) \times (k_1 2^{2k'_2+2}+1)})$.
    Using only reusable cops, the cops can reduce the size of the robber component with a barrier.
    However, the size remains at least $(k_1-1)\cdot (k_1 2^{2k'_2+1}+1)$, since the robber can choose the larger induced component. 
    When a reusable cop is reused before a non-reusable cop was used to build another barrier, the reusable barrier breaks down and the robber as an escape strategy. 
    Using non-reusable cops, the cops can build another wall to reduce the size of the robber component to $(k_1-1)\cdot (k_1 2^{2k'_2}+1)$. 
    Again, the robber chooses the larger induced component. 
    But now the remaining game is $\CR_{(k'_1,k'_2-k_1+2)}(G_{(k_1-1) \times (k_1 2^{2k'_2}+1)})$ and we have $k'_2 \geq k'_2-k_1+2$.
    Thus, by the inductive hypothesis the robber has a winning strategy for the remaining game. 
\end{proof}

Second, we show the advantages of mere capacity:
When the base graph is chosen as a complete graph, the robber can choose any edge independently of the choice of vertices by the cops and the only possibility to win for the cops is to have sufficient capacity.
In this case, capacity is more valuable than reusability.

\begin{lemma} \label{lem:capacity:advantage}
    For all \(k_1,k_2,k_1',k_2'\geq 0\), if \(k_1+k_2>k_1'+k_2'\), then \(\mathsf{C}^{(k_1,k_2)}\not\preceq\mathsf{C}^{(k_1',k_2')}\).
\end{lemma}
\begin{proof}
    In the game $\CR_{(k_1,k_2)}(K_{k_1+k_2})$, the cops have a winning strategy just by covering all base vertices.
    In contrast, in the game $\CR_{(k'_1,k'_2)}(K_{k_1+k_2})$ the robber has a winning strategy. Whenever a cop is picked up there is one edge that is not incident to a cop and thus yields a safe escape for the robber.
\end{proof}

Third, we treat the special case of a single requantifiable variable:
Intuitively, at least two reusable cops are needed to move a barrier for an arbitrarily large distance in a base graph.
When only one single reusable cop is available, the distance that can be covered by a barrier of cops is bounded by $2k_1+1$ because for every other move a non-reusable cop must be used.
The \emph{perfect binary tree of depth $d$} is the binary tree $B^d$ such that all interior vertices have two children and all leaves have the same depth.
\begin{lemma} \label{lem:special_case}
    For all $k_2 \geq 1$ it holds that
    $\mathsf{C}^{(1,k_2)} \not\preceq \mathsf{C}^{(0,k'_2)}$ if and only if $k'_2 \leq 2 k_2$.
\end{lemma}
\begin{proof}
    In the game $\CR_{(1,k_2)}(B^{2 k_2})$, the cops have a winning strategy by alternately using non-reusable cops and the reusable cop.
    In the game $\CR_{(0,2k_2)}(B^{2 k_2})$ the robber has a winning strategy as the non-reusable cops are exhausted before the robber is caught.
    This yields $\mathsf{C}^{(1,k_2)} \not\preceq \mathsf{C}^{(0,2k_2)}$. For $k'_2 \leq 2 k_2$ we get $\mathsf{C}^{(1,k_2)} \not\preceq \mathsf{C}^{(0,k'_2)}$ since the robber wins with the same strategy in $\CR_{(0,k'_2)}(B^{2 k_2})$.
    For $k'_2 > 2k_2$, let (S) have a winning strategy for $\BP_{(1,k_2)}(G,H)$.
    Then (S) has a winning strategy for $\BP^{2k_2+1}_{(1,k_2)}(G,H)$ since consecutive moves involving the pebble pair $x_1$ can be replaced by a single move instead.
    The winning strategy for (S) in $\BP^{2k_2+1}_{(1,k_2)}(G,H)$ directly yields a winning strategy for (S) in $\BP_{(0,k'_2)}(G,H)$:
    For every pebble pair played by (S) in $\BP^{2k_2+1}_{(1,k_2)}(G,H)$, the player (S) can use a new pair in $\BP_{(0,k'_2)}(G,H)$.
\end{proof}

With the previous lemmas we can determine the relation of the logics $\mathsf{C}^{(k_1,k_2)}$ and $\mathsf{C}^{(k'_1,k'_2)}$ for any given combination of parameters.

\begin{theorem} \label{thm:hierarchy}
    For all $k_1,k_2 \in \mathbb{N}$ and $k'_1,k'_2 \in \mathbb{N}$ with $k_1+k_2,k'_1+k'_2 \geq 2$ it holds that
    $\mathsf{C}^{(k_1,k_2)} \prec \mathsf{C}^{(k'_1,k'_2)}$ if and only if one of the following assertions holds:
    \begin{enumerate}
        \item $k_1 < k'_1$ and $k_1+k_2 \leq k'_1+k'_2$,
        \item $k_1 \leq k'_1$ and $k_1+k_2 < k'_1+k'_2$, or
        \item $k_1=1$, $k'_1 = 0$, and $k'_2 > 2k_2$.
    \end{enumerate}
    Furthermore, it holds that $\mathsf{C}^{(k_1,k_2)} \equiv \mathsf{C}^{(k'_1,k'_2)}$ if and only if~$(k_1,k_2) = (k'_1,k'_2)$.
\end{theorem}

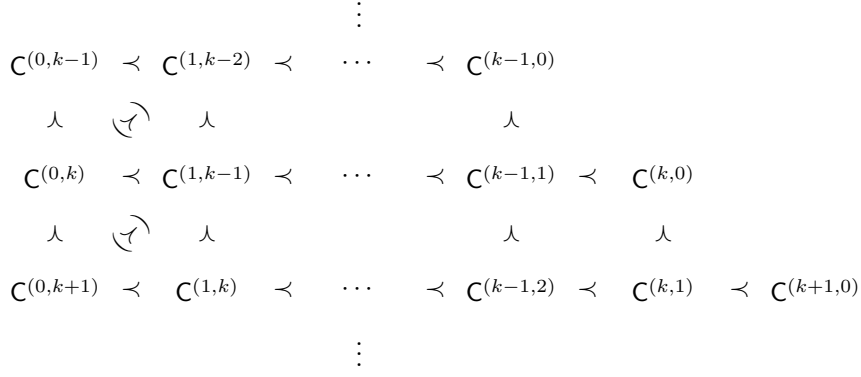
\begin{figure}
    \centering
    \begin{tikzpicture}
        \node at (0, 3) {$\vdots$};
        \node (0, k-1) at (-4,2.25) {$\mathsf{C}^{(0,k-1)}$};
        \node at (-3 , 2.25) {$\prec$};
        \node (1, k-2) at (-2,2.25) {$\mathsf{C}^{(1, k-2)}$};
        \node at (-1 , 2.25) {$\prec$};
        \node at (0, 2.25) {$\cdots$};
        \node at (1 , 2.25) {$\prec$};
        \node (k-1, 0) at (2, 2.25) {$\mathsf{C}^{(k-1,0)}$};

        \node [rotate=270] at (-4, 1.5) {$\prec$};
        \node [rotate=270] at (-2, 1.5) {$\prec$};
        \node [rotate=270] at (2, 1.5) {$\prec$};

        \node (0, k) at (-4,0.75) {$\mathsf{C}^{(0,k)}$};
        \node at (-3 , 0.75) {$\prec$};
        \node (1, k-1) at (-2,0.75) {$\mathsf{C}^{(1, k-1)}$};
        \node at (-1 , 0.75) {$\prec$};
        \node at (0, 0.75) {$\cdots$};
        \node at (1 , 0.75) {$\prec$};
        \node (k-1, 1) at (2 ,0.75) {$\mathsf{C}^{(k-1,1)}$};
        \node at (3 , 0.75) {$\prec$};
        \node (k, 0) at (4, 0.75) {$\mathsf{C}^{(k,0)}$};

        \node [rotate=270] at (-4, 0) {$\prec$};
        \node [rotate=270] at (-2, 0) {$\prec$};
        \node [rotate=270] at (2, 0) {$\prec$};
        \node [rotate=270] at (4, 0) {$\prec$};

        \node (0, k+1) at (-4,-0.75) {$\mathsf{C}^{(0,k+1)}$};
        \node at (-3 , -0.75) {$\prec$};
        \node (1, k) at (-2,-0.75) {$\mathsf{C}^{(1, k)}$};
        \node at (-1 , -0.75) {$\prec$};
        \node at (0, -0.75) {$\cdots$};
        \node at (1 , -0.75) {$\prec$};
        \node (k-1, 1) at (2,-0.75) {$\mathsf{C}^{(k-1,2)}$};
        \node at (3 , -0.75) {$\prec$};
        \node (k, 0) at (4,-0.75) {$\mathsf{C}^{(k,1)}$};
        \node at (5 , -0.75) {$\prec$};
        \node (k, 0) at (6, -0.75) {$\mathsf{C}^{(k+1,0)}$};

        \node at (0, -1.5) {$\vdots$};

        \node [rotate=225] at (-3, 1.5) {$(\prec)$};
        \node [rotate=225] at (-3, 0) {$(\prec)$};
    \end{tikzpicture}
    \caption{A graphical depiction of the hierarchy of logics from \Cref{thm:hierarchy}}
    \label{fig:hierarchy}
\end{figure}

This settles the question of how the use of non-requantifiable variables affects the expressive power of the logic.
The resulting hierarchy is depicted in \Cref{fig:hierarchy}.

For the investigation of the logic $\mathsf{C}^{(k_1,k_2)}$ for fixed parameters, it is also of interest how the non-requantifiable variables behave in concrete formulas of the logic.
This relates closely to asking whether there are normal forms for the logic $\mathsf{C}^{(k_1,k_2)}$ with respect to reusability.
We give a precise answer to this question that rules out many such normal forms.

\begin{definition}
    We construct the graph $\dot{G}_{h \times \ell}$ from the grid graph $G_{h \times \ell}$ by adding one additional vertex $b$,
    which we call \textit{bridge vertex}, and the edges $\{\{v_{i, \lfloor \frac{\ell}{2} \rfloor}, b\} : i \in [h]\} \cup \{\{b, v_{i, \lfloor \frac{\ell}{2}+1 \rfloor}\} : i \in [h]\}$ to $G_{h \times \ell}$.
    Using this modified grid, we define the following base graph:

    For $\ell \geq 2, h\geq 1$ and $d \geq 1$ we obtain the graph $B^d_{h \times \ell}$ by replacing every vertex of a perfect binary tree~$B^d$ of depth $d$
    by a grid $\dot{G}_{h \times \ell}$ and connect adjacent grids row-wise (see \Cref{fig:binary:tree}).
\end{definition}

\begin{theorem} \label{thm:normal:form}
    For all $k_1, k_2 \geq 1$ and $r \geq 1$ there exist graphs $G$ and $H$ such that (S) has a winning strategy for $\BP_{(k_1,k_2)}(G,H)$
    and in every winning strategy (S) must (re)use every reusable pebble pair at least $r$ times (once if $k_1=1$) before using a new non-reusable pebble pair.
\end{theorem}
\begin{proof}
    For \(G\) and \(H\), we again choose the twisted and untwisted CFI graphs over an appropriate base graph, and analyze winning strategies in the cops-and-robbers game.

    For $k_1 > 1$, we consider the game $\CR_{(k_1,k_2)}(B^{k_2+1}_{(k_1 -1) \times 2r})$, see \Cref{fig:binary:tree}.
    The cops have the following winning strategy: First, they build a barrier in the root grid (behind the bridge vertex) by occupying one full column using $k_1$ reusable cops.
    The barrier can then be moved towards the robber using the additional reusable cop. 
    When the barrier reaches the last column, the two subtrees induced by the children of the root grid are disconnected components with respect to the cops.
    Thus, the robber has to choose an edge in one of these components to escape to.
    The cops move the barrier into the corresponding subtree, which essentially results in the game $\CR_{(k_1,k_2)}(B^{k_2}_{(k_1-1)\times2r})$.
    In every grid of the tree, the cops encounter a bridge vertex that can be covered by one of the $k_2$ non-reusable cops.
    The game continues inductively for $\Omega(r k_2)$ rounds, until the cops use the last remaining non-reusable cop to cover the bridge vertex in a leaf grid.
    The barrier can be moved to the end of the grid and the robber will be caught.
    Now assume a new non-reusable cop $y$ has been used before all reusable cops have been (re)used $r$ times at some point of the game.
    Then the barrier was not moved out of the current grid at this point, since all reusable cops have to be moved at least $r$ times to achieve this.
    Hence, the robber has not chosen a new subtree so far and can pick an edge in a subtree that does not contain $y$.
    The cops need to move the barrier into that subtree and use non-reusable cops for the bridge vertices.
    Since $y$ was used in another subtree, at some point there will be no non-reusable cops left to cover a bridge vertex and the barrier
    cannot be moved further without breaking down. Thus, the robber can escape indefinitely.

    For $k_1 = 1$ we consider the game $\CR_{(1,k_2)}(B^{2 k_2})$. 
    The cops have the following winning strategy: 
    First, the reusable cop $x_1$ is placed in the root node. 
    This disconnects the subtrees induced by the two children of the root node for the robber, and the robber has to choose an edge in one of the subtrees of depth $2k_2-1$.
    Accordingly, a non-reusable cop $y_1$ is placed on the node inducing that subtree, which again disconnects two subtrees of depth $2 k_2 -2$.
    The cop $x_1$ can be picked up again from the root node to be placed on the corresponding subtree.
    Inductively, the cops alternately use non-reusable cops $y_j$ and the reusable cop $x_1$ to cover the next child node.
    After $2k_2$ moves, the induced subtree is of depth $0$ and the robber is caught in the edge to a leaf node. 
    If two non-reusable cops are used consecutively, similar to the case $k_1 > 1$, there are no non-reusable cops left at depth $2k_2-1$ and the remaining reusable cop does not suffice to catch the robber. 
\end{proof}

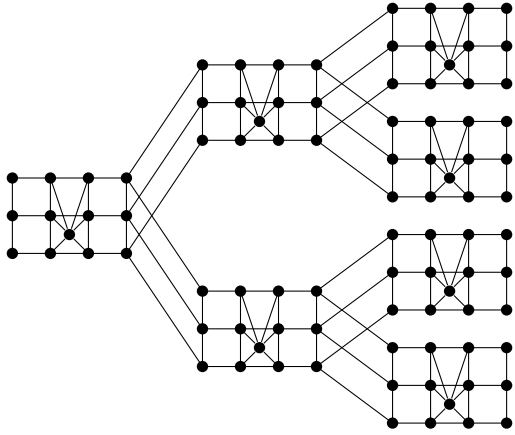
\begin{figure}
    \centering
    \begin{tikzpicture}[scale=0.5]
        \def\maxX{3} % must be odd
        \def\maxY{2}

        % ROOT 

        \foreach \x in {0,...,\maxX}{
                \foreach \y in {0,...,\maxY}{
                        \node[circle, fill = black, inner sep=1.5pt] (\x,\y) at (\x,\y) {};
                    }
            }
        \foreach \x in {0,...,\maxX}{
                \draw (\x, 0) to (\x, \maxY);
            }
        \foreach \y in {0,...,\maxY}{
                \draw (0,\y) to (\maxX, \y);
            }
        \node [circle, fill = black, inner sep=1.5pt] (b) at (\maxX / 2, 0.5) {};
        \foreach \y in {0,...,\maxY}{
                \draw (b) to ({(\maxX /2) - 0.5}, \y);
                \draw (b) to ({(\maxX /2) + 0.5}, \y);
            }

        % UP

        \foreach \x in {0,...,\maxX}{
                \foreach \y in {0,...,\maxY}{
                        \node[circle, fill = black, inner sep=1.5pt] ({5+\x},{3+\y}) at ({5+\x},{3+\y}) {};
                    }
            }
        \foreach \x in {0,...,\maxX}{
                \draw ({5+\x}, 3) to ({5+\x}, {3+\maxY});
            }
        \foreach \y in {0,...,\maxY}{
                \draw (5,{3+\y}) to ({5+\maxX}, {3+\y});
            }
        \node [circle, fill = black, inner sep=1.5pt] (b_up) at ({5 +(\maxX / 2)}, {3.5}) {};
        \foreach \y in {0,...,\maxY}{
                \draw (b_up) to ({5 + (\maxX /2) - 0.5}, {3+\y});
                \draw (b_up) to ({5+ (\maxX /2) + 0.5}, {3+\y});
            }

        % UP_UP

        \foreach \x in {0,...,\maxX}{
                \foreach \y in {0,...,\maxY}{
                        \node[circle, fill = black, inner sep=1.5pt] ({10+\x},{4.5+\y}) at ({10+\x},{4.5+\y}) {};
                    }
            }
        \foreach \x in {0,...,\maxX}{
                \draw ({10+\x}, 4.5) to ({10+\x}, {4.5+\maxY});
            }
        \foreach \y in {0,...,\maxY}{
                \draw (10,{4.5+\y}) to ({10+\maxX}, {4.5+\y});
            }
        \node [circle, fill = black, inner sep=1.5pt] (b_up_up) at ({10 +(\maxX / 2)}, {5}) {};
        \foreach \y in {0,...,\maxY}{
                \draw (b_up_up) to ({10 + (\maxX /2) - 0.5}, {4.5+\y});
                \draw (b_up_up) to ({10+ (\maxX /2) + 0.5}, {4.5+\y});
            }

        % UP_DOWN

        \foreach \x in {0,...,\maxX}{
                \foreach \y in {0,...,\maxY}{
                        \node[circle, fill = black, inner sep=1.5pt] ({10+\x},{1.5+\y}) at ({10+\x},{1.5+\y}) {};
                    }
            }
        \foreach \x in {0,...,\maxX}{
                \draw ({10+\x}, 1.5) to ({10+\x}, {1.5+\maxY});
            }
        \foreach \y in {0,...,\maxY}{
                \draw (10,{1.5+\y}) to ({10+\maxX}, {1.5+\y});
            }
        \node [circle, fill = black, inner sep=1.5pt] (b_up_up) at ({10 +(\maxX / 2)}, {2}) {};
        \foreach \y in {0,...,\maxY}{
                \draw (b_up_up) to ({10 + (\maxX /2) - 0.5}, {1.5+\y});
                \draw (b_up_up) to ({10+ (\maxX /2) + 0.5}, {1.5+\y});
            }

        % DOWN

        \foreach \x in {0,...,\maxX}{
                \foreach \y in {0,...,\maxY}{
                        \node[circle, fill = black, inner sep=1.5pt] ({5+\x},{-3+\y}) at ({5+\x},{-3+\y}) {};
                    }
            }
        \foreach \x in {0,...,\maxX}{
                \draw ({5+\x}, -3) to ({5+\x}, {-3+\maxY});
            }
        \foreach \y in {0,...,\maxY}{
                \draw (5,{-3+\y}) to ({5+\maxX}, {-3+\y});
            }
        \node [circle, fill = black, inner sep=1.5pt] (b_down) at ({5 +(\maxX / 2)}, {-2.5}) {};
        \foreach \y in {0,...,\maxY}{
                \draw (b_down) to ({5 + (\maxX /2) - 0.5}, {-3+\y});
                \draw (b_down) to ({5+ (\maxX /2) + 0.5}, {-3+\y});
            }

        % DOWN_UP

        \foreach \x in {0,...,\maxX}{
                \foreach \y in {0,...,\maxY}{
                        \node[circle, fill = black, inner sep=1.5pt] ({10+\x},{-1.5+\y}) at ({10+\x},{-1.5+\y}) {};
                    }
            }
        \foreach \x in {0,...,\maxX}{
                \draw ({10+\x}, -1.5) to ({10+\x}, {-1.5+\maxY});
            }
        \foreach \y in {0,...,\maxY}{
                \draw (10,{-1.5+\y}) to ({10+\maxX}, {-1.5+\y});
            }
        \node [circle, fill = black, inner sep=1.5pt] (b_down_up) at ({10 +(\maxX / 2)}, {-1}) {};
        \foreach \y in {0,...,\maxY}{
                \draw (b_down_up) to ({10 + (\maxX /2) - 0.5}, {-1.5+\y});
                \draw (b_down_up) to ({10+ (\maxX /2) + 0.5}, {-1.5+\y});
            }

        % DOWN_DOWN

        \foreach \x in {0,...,\maxX}{
                \foreach \y in {0,...,\maxY}{
                        \node[circle, fill = black, inner sep=1.5pt] ({10+\x},{-1.5+\y}) at ({10+\x},{-4.5+\y}) {};
                    }
            }
        \foreach \x in {0,...,\maxX}{
                \draw ({10+\x}, -4.5) to ({10+\x}, {-4.5+\maxY});
            }
        \foreach \y in {0,...,\maxY}{
                \draw (10,{-4.5+\y}) to ({10+\maxX}, {-4.5+\y});
            }
        \node [circle, fill = black, inner sep=1.5pt] (b_down_down) at ({10 +(\maxX / 2)}, {-4}) {};
        \foreach \y in {0,...,\maxY}{
                \draw (b_down_down) to ({10 + (\maxX /2) - 0.5}, {-4.5+\y});
                \draw (b_down_down) to ({10+ (\maxX /2) + 0.5}, {-4.5+\y});
            }

        % Connect root

        \draw (3,0) to (5,3);
        \draw (3,1) to (5,4);
        \draw (3,2) to (5,5);

        \draw (3,0) to (5,-3);
        \draw (3,1) to (5,-2);
        \draw (3,2) to (5,-1);

        % connect up

        \draw (8,3) to (10,4.5);
        \draw (8,4) to (10,5.5);
        \draw (8,5) to (10,6.5);

        \draw (8,3) to (10,1.5);
        \draw (8,4) to (10,2.5);
        \draw (8,5) to (10,3.5);

        % connect down

        \draw (8,-3) to (10,-1.5);
        \draw (8,-2) to (10,-0.5);
        \draw (8,-1) to (10,0.5);

        \draw (8,-3) to (10,-4.5);
        \draw (8,-2) to (10,-3.5);
        \draw (8,-1) to (10,-2.5);
    \end{tikzpicture}
    \caption{A drawing of the graph $B^3_{3 \times 4}$}
    \label{fig:binary:tree}
\end{figure}

Intuitively, this result can be translated into the language of logic as follows:
In order to achieve the full expressive power of the logic $\mathsf{C}^{(k_1,k_2)}$, it is indispensable to use rather complicated quantification patterns.
That is, all requantifiable variables $x_i$ must be requantified arbitrarily often before a new non-requantifiable variable $y_j$ is quantified.
More precisely, by~\Cref{thm:characterization} we have the following:

\begin{corollary}\label{cor:no:normalform}
    For all $k_1, k_2 \geq 1$ and $r \geq 1$ there exist graphs $G, H$ such that $G$ and $H$ are not $\mathsf{C}^{(k_1,k_2)}$-equivalent and for every formula $\varphi \in \mathsf{C}^{(k_1,k_2)}$
    that distinguishes $G$ and $H$ the following holds:
    There exists a sequence of subformulas $\exists^{\geq n_1} y_{j_1} \psi_1, \dotsc, \exists^{\geq n_{k_2}} y_{j_{k_2}} \psi_{k_2}$ of $\varphi$ such that
    $\qr(\psi_{k_2}) = k_1$ and
    the formula $\psi_{\ell+1}$ is a subformula of $\psi_{\ell}$ with
    $$\qr(\psi_{\ell}) \geq \begin{cases} \qr(\psi_{\ell+1}) + k_1r & \text{ if } k_1 >1 \\
              \qr(\psi_{\ell+1}) +1     & \text{ if } k_1 =1\end{cases}$$ for $\ell \in [k_2-1]$.
    Moreover, between the quantifications $\exists^{\geq n_\ell} y_{j_\ell} \psi_\ell$ and $\exists^{\geq n_{\ell+1}} y_{j_{\ell+1}} \psi_{\ell+1}$ all requantifiable variables have to be requantified $r$ times (once if $k_1=1$) in \(\psi_\ell\).
\end{corollary}

The necessity of this pattern of (re)quantification rules out the possibility of various normal forms with respect to requantification for $\mathsf{C}^{(k_1,k_2)}$ one might have hoped to have.
In particular, it is not sufficient to quantify all non-requantifiable variables directly one after the other.

Regarding classical logics without restricted requantification, \Cref{thm:hierarchy} and \Cref{cor:no:normalform} yield that for $k_1, k_2 \geq 1$ the power of $\mathsf{C}^{(k_1,k_2)}$ to distinguish graphs is not identical to that of $\mathsf{C}^k, \mathsf{C}_r,$ or $\mathsf{C}^k_r$ for all $k,r \in \mathbb{N}$. 

\section{Space Complexity}\label{sec:space:compl}
\newcommand{\dotcup}{\mathbin{\dot\cup}}
In this section we investigate the space complexity of deciding whether two given graphs are $\mathsf{C}^{(k_1,k_2)}$-equivalent.
In principle, this can be achieved by testing $\owl_{(k_1,k_2)}^{(\infty)}(G) = \owl_{(k_1,k_2)}^{(\infty)}(H)$ by \Cref{thm:characterization}.
However, a naive implementation of $(k_1,k_2)$-OWL requires space $\Omega(n^{k_1+k_2})$ and hence provides no improvement compared to the situation with unrestricted reusability.
We seek to improve the space complexity to \(O\bigl(n^{k_1}\log n\bigr)\) when both requantifiable and non-requantifiable variables are involved. Here the $O$ notation hides factors depending on $k_1$ and $k_2$ but not on $n$.

To achieve this, we observe that the color \(\owl_{(k_1,k_2)}^{(r)}(G,\overline{\alpha})\)
only depends on the colors \(\owl_{(k_1,k_2)}^{(s)}(G,\overline{\beta})\)
with \(s<r\) where \(\beta|_{[y_{k_2}]}\) is an extension of \(\alpha|_{[y_{k_2}]}\).
This allows us to compute these colorings, while only ever remembering colors of assignments
with a few distinct \([y_{k_2}]\)-parts.
Moreover, we show that the \((k_1,k_2)\)-dimensional oblivious Weisfeiler-Leman algorithm
can equivalently be implemented by alternatingly refining with respect to the reusable dimensions until the coloring stabilizes, and refining with respect to the first unassigned non-requantifiable variable.
We use this to show that the iteration number of \((k_1,k_2)\)-OWL is at most \((k_2+1)n^{k_1}-1\).

\newcommand{\finereq}{\preceq}
\newcommand{\coarsereq}{\succeq}
\begin{definition}
    For colorings \(\chi\) and \(\chi'\) on a set \(S\),
    we say that \(\chi\) \emph{refines}
    \(\chi'\), written \(\chi\finereq\chi'\), if every \(\chi'\)-color class is a union of \(\chi\)-color classes.
    If \(\chi\finereq\chi'\) and \(\chi\coarsereq\chi'\), we write \(\chi\equiv\chi'\).
\end{definition}

In the context of colorings, it is natural to understand the oblivious Weisfeiler-Leman algorithm as
a \emph{refinement operator}, i.e., a function that maps every coloring to a refined coloring.
To make this formal, we define for every coloring \(\chi\) on assignments \(\alpha \colon [x_{k_1},y_{k_2}]\rightharpoonup V(G)\) the OWL-refinement
\begin{align*}
\owlref_{(k_1,k_2)}(\chi)(\alpha)\coloneqq \Bigl(\chi(\alpha),
	&\{\!\!\{ \chi(\alpha[x_i/w]) : w\in V(G) \}\!\!\}_{i\in [k_1]},\\
	&\{\!\!\{ \chi(\alpha[y_j/w]) : w\in V(G) \}\!\!\}_{j\in J(\alpha)}\Bigr).
\end{align*}
This refinement is precisely the refinement that is applied by OWL in each iteration. In particular, applying it \(r\) times to the initial coloring by atomic types
yields precisely the \(r\)-round OWL-coloring. That is,
\[\owlref_{(k_1,k_2)}^{(r)}(\operatorname{atp}_{k_1+k_2}(G))=\owl_{(k_1,k_2)}^{(r)}(G).\]
Note that OWL-refinement is \emph{monotone} in the sense that for all colorings \(\chi\) and \(\chi'\) with the property
\(\chi\finereq\chi'\) it also holds that \(\owlref_{(k_1,k_2)}(\chi)\finereq\owlref_{(k_1,k_2)}(\chi')\).

In order to space-efficiently deal with these refinements, we want to separate the refinements
with respect to reusable dimensions from those with respect to non-reusable dimensions.
To do this, note that the definition of \(\owlref_{(k_1,k_2)}\) still makes sense
for colorings assignments \(\alpha \colon [x_{k_1'},y_{k_2'}]\rightharpoonup V(G)\)
for \(k_1'\geq k_1\) or \(k_2'\geq k_2\), where the refinement just refines
with respect to some but not all of the dimensions.

Moreover, in order to handle the distinguishing power of \((k_1,k_2)\)-OWL on two different graphs,
we note that we can also simultaneously apply these refinement operators to colorings on assignments over two different graphs.

With this terminology at hand, we can clarify the intuition we may gain from \Cref{sec:role:of:reuse} regarding the employment of non-reusability.
That is, in order to distinguish graphs, it suffices to alternatingly refine with respect
to all requantifiable variables and a non-requantifiable variable.
\begin{lemma}\label{lem:refinement:nf}
    For all \(k_1+k_2\geq 1\), we have
    \[\owl_{(k_1,k_2)}^{(\infty)}(G)\equiv(\owlref_{(k_1,0)}^{(\infty)} \circ \owlref_{(0,k_2)})^{(k_2)}(\owl_{(k_1,0)}^{(\infty)}(G))\]
\end{lemma}
\begin{proof}
    Because \(\owl_{(k_1,k_2)}\) is a finer refinement than \(\owlref_{(k_1,0)}\) and than \(\owlref_{(0,k_2)}\),
    the direction \(\finereq\) is immediate.
    For the other direction,
    set
    \[\chi_r\coloneqq
        \left(\owlref_{(k_1,0)}^{(\infty)}\circ
        \owlref_{(0,k_2)}\right)^{(r)}\left(\owl_{(k_1,0)}^{(\infty)}(G)\right).\]

    We use the bijective pebble game and show, by induction on \(r\), that for every \((k_1,k_2)\)-configuration
    \((\alpha,\beta)\) over \(G\) with \(\lvert\operatorname{dom}(\alpha)\cap[y_{k_2}]\rvert=k_2-r\)
    such that \(\alpha\) and \(\beta\) have equal \(\chi_r\)-colors,
    (D) has a winning strategy in the game  \(\BP_{(k_1,k_2)}(G,G)\) with initial position \((\alpha,\beta)\). This then implies the equality of \(\owlref_{(k_1,k_2)}^{(\infty)}(\chi)\)-colors.

    For \(r=0\), the colors are precisely the colors computed by \((k_1,k_2)\)-OWL.
    Thus, (D) has a winning strategy by \Cref{thm:characterization}.

    Now, assume the claim is true for some \(r\) and let \((\alpha,\beta)\) be a \((k_1,k_2)\)-configuration
    with \(\lvert\operatorname{dom}(\alpha)\cap[y_{k_2}]\rvert=k_2-(r+1)\) such that \(\alpha\) and \(\beta\)
    have equal \(\chi_{r+1}\)-colors.
    By the construction of \((k_1,0)\)-OWL refinement, (D) can preserve the equality of \(\chi_{r+1}\)-colors
    as long as (S) picks up reusable pebble pairs. When (S) picks up a non-reusable pebble pair,
    (D) can play such that the resulting positions have the same \(\chi_r\)-colors.
    But then, (D) has a winning strategy by the induction hypothesis.
\end{proof}

Because classical \(\owl_{(k_1,0)}\)-refinement stabilizes after at most \(n^{k_1}-1\) rounds,
this scheme yields an upper bound on the iteration number of the oblivious Weisfeiler-Leman algorithm.
\begin{corollary}\label{cor:iteration:number}
    The sequence of colorings \(\owl_{(k_1,k_2)}^{(r)}\) computed on a graph \(G\) stabilizes after at most \((k_2+1)n^{k_1}-1\) rounds.
\end{corollary}

We will now turn to the computation of the OWL-colorings.
Because the names of the OWL-colors consist of nested multisets, they can become exponentially long.
The usual way to deal with this is to either replace after each iteration round all color names
by numbers of logarithmic length, or to not compute the colors at all but only
consider the order on the variable assignments induced by the lexicographic ordering of their OWL-colors.
We will switch between these two viewpoints depending on suitability to the task at hand.
Accordingly, we use two different encodings of the colorings we deal with.
Consider a coloring \(\chi\colon M\to C\) and an order \(\leq\) on the set of colors \(C\).
We say that an algorithm is given \emph{oracle access to the ordering of \(\chi\)-colors}
if the algorithm has access to a function that, given two elements \(m,m'\in M\),
returns whether \(\chi(m)\leq \chi(m')\).
For the second way that our algorithms interact with colorings,
we call a coloring \(\chi'\colon M\to[|M|]\) a \emph{normalization of \(\chi\)}
if for all \(m,m'\in M\) we have \(\chi(m)\leq \chi(m')\) if and only if \(\chi'(m)\leq \chi'(m')\).
Now, we say that an algorithm is given a \emph{function table for \(\chi\)} if for some
normalization \(\chi'\) of \(\chi\) the algorithm is given an array \(A\) with \(A[m]=\chi'(m)\)
for all \(m\in M\) suitably encoded as numbers in \([|M|]\). Similarly,
we say that an algorithm \emph{computes a function table for \(\chi\)}
if it outputs such an array. Note that a function table can be stored
in space \(|M|\cdot\lceil\log_2|M|\rceil\in O(|M|\log |M|)\).

The main technical tool needed for the implementation of \(\owlref_{(k_1,k_2)}\)-refinements,
is the ability to compare multisets of previously computed colors when given oracle access
to a function comparing these previous colors.
\begin{lemma}\label{lem:lex:order:multisets}
	Given a natural number $n$ in unary, oracle access to a total order \(\preceq\) on \([n]\), and two multisets \(M\) and \(M'\) on \([n]\) of order at most \(n\),
	the lexicographic order of \(M\) and \(M'\) can be decided in logarithmic space using quadratic time.
	Also, the lexicographical order of tuples of colors can be computed in logarithmic space.
\end{lemma}
\begin{proof}
	Consider two multisets
	\(M=\{\!\!\{s_1,\dots,s_n\}\!\!\}\) and \(M'=\{\!\!\{s_1',\dots,s_n'\}\!\!\}\).
	Note that we have enough space to store a constant number of elements of \([n]\).

	For a number \(i\in[n]\), we denote the number of occurrences of \(i\) in \(M\) or \(M'\)
	by \(M(i)\) and \(M'(i)\) respectively. Note that these numbers can be computed
	in logarithmic space and linear time by simply comparing \(i\) to all elements in either set.

	We start by finding the minimal element \(m_1\) of \(M\) and \(m_1'\) of \(M'\).
	If \(m_1\neq m_1'\), we return their order. Otherwise, if \(M(m_1)\neq M'(m_1)\),
	we return this order.

	Thus, assume \(m_1=m_1'\) and that they occur in both multisets the same number of times.
	Next, we find the second-smallest elements \(m_2\) and \(m_2'\) of both sets,
	and can now forget about \(m_1\) and \(m_1'\). We again compare \(m_2\) and \(m_2'\)
	and their number of occurrences and possibly return the order accordingly.
	Iteratively, we only need to remember the \(i\)-th smallest elements
	to find the \((i+1)\)-th smallest elements, and we iteratively
	compare these elements and their number of occurrences.
\end{proof}

This allows us to compute the order of \(\owlref_{(k_1,k_2)}(\chi)\)-colors in logarithmic space
when we are given oracle access to the order of \(\chi\)-colors.
Using the bound on the iteration number of \((k_1,k_2)\)-OWL from \Cref{cor:iteration:number},
and the fact that we can perform one iteration using only logarithmic additional space,
we immediately obtain an algorithm that can compare \(\owl_{(k_1,k_2)}^{(\infty)}\)-colors
using space at most \(O(n^{k_1}\log n)\), where we again dropped multiplicative factors
depending on \(k_1\) and \(k_2\).
However, this naive implementation will not run in polynomial time.
Indeed, because there are already \(n^{k_1+k_2}\) many variable assignments,
we do not have enough space to store even the (order of) colors computed in the previous round.
Instead, this naive algorithm recomputes polynomially many previous colors in every step,
which leads to a polynomially branching algorithm with exponential running time.

To remedy this, we make full use of the scheme from \Cref{lem:refinement:nf}.
While performing \(\owl_{(k_1,0)}\)-refinements, we are able to store
a function table with the previously computed colors for all assignments with the same
\([y_{k_2}]\)-part. This allows us to perform a full \(\owl_{(k_1,0)}^{(\infty)}\)-refinement
in polynomial time and the required space.
Only when performing one of the \(k_2\) many \(\owl_{(0,k_2)}\)-refinement steps
do we need to consider variable assignments with different \([y_{k_2}]\)-parts.
In this latter case, we cannot circumvent needing to compute colors for
these assignments polynomially many times. While this does again
lead to a polynomially branching algorithm, the depth of this branching
is bounded by \(k_2\), which leads to a polynomial running time increase
of \(n^{k_2}\).

For a fixed assignment \(\eta\colon[y_{k_2}]\rightharpoonup V(G)\),
we denote by \([[x_{k_1},y_{k_2}]\rightharpoonup V(G)]_{\eta}\)
the set of assignments \(\alpha\colon[x_{k_1},y_{k_2}]\rightharpoonup V(G)\)
whose \([y_{k_2}]\)-part is \(\eta\).
Because we only ever need to compare the colors of two assignments at a time,
it will always be sufficient to compute the OWL-coloring on sets of the form
\[[[x_{k_1},y_{k_2}]\rightharpoonup V(G)]_{\eta_G}\dotcup[[x_{k_1},y_{k_2}]\rightharpoonup V(H)]_{\eta_H}\]
for a \((0,k_2)\)-configuration \((\eta_G,\eta_H)\) over \(G\) and \(H\).
When restricting ourselves to assignments in such a set,
the coloring computed by \((k_1,0)\)-OWL can be computed as usual:
\begin{lemma}\label{lem:space:owl:reusable}
    Let \(k_1+k_2\geq 1\), \(G, H\) be graphs and \((\eta_G,\eta_H)\) be a \((0,k_2)\)-configuration over \(G, H\).
    Given a function table for a coloring \(\chi\) on \([[x_{k_1},y_{k_2}]\rightharpoonup V(G)]_{\eta_G}\dotcup[[x_{k_1},y_{k_2}]\rightharpoonup V(H)]_{\eta_H}\), 
    a function table for \(\owlref_{(k_1,0)}(\chi)\) can be computed in time \(n^{O(k_1)}\) and space \(O(k_1n^{k_1}\log n)\).
\end{lemma}
\begin{proof}
    Note that
	\(\bigl|[[x_{k_1},y_{k_2}]\rightharpoonup V(G)]_{\eta_G}\dotcup[[x_{k_1},y_{k_2}]\rightharpoonup V(H)]_{\eta_H}\bigr|=2(n+1)^{k_1}\),
	which means that we can store a function tables for \(\chi\) and \(\owlref_{(k_1,0)}(\chi)\) in space
	\[O\left(2(n+1)^{k_1}\cdot\log\left(2(n+1)^{k_1}\right)\right)=O\left(k_1 n^{k_1}\log n\right).\]
	In addition to these two function tables, we will only need logarithmic space.

    In order to compute the function table for
    \(\owlref_{(k_1,0)}(\chi)\), we need to refine the coloring \(\chi\) with respect to
    the multisets
    \[\{\!\!\{\chi(\alpha[x_i/w]) : w\in V(G)\}\!\!\}\quad\text{or}\quad
        \{\!\!\{\chi(\alpha[x_i/w]) : w\in V(H)\}\!\!\}\]
    for all \(i\in[k_1]\). By using the function table for \(\chi\) as an oracle,
    we can compare these multisets in logarithmic additional space
    using \Cref{lem:lex:order:multisets}.

    This allows us to compare \(\owlref_{(k_1,0)}(\chi)\)-colors in the required space.
    Now, we simply start to compare each variable assignment \(\alpha\) with all other assignments
    and count the number of assignments whose color is less than or equal to \(\alpha\).
    Then, we use this count as the new color of \(\alpha\) and insert it into our function table.
\end{proof}

By applying \Cref{lem:space:owl:reusable} repeatedly until the coloring stabilizes,
and only ever storing the function table from the previous and current iteration round
we get the following:
\begin{corollary}\label{cor:space:owl:reusable}
    Let \(k_1+k_2\geq 1\), \(G\) and \(H\) be graphs, and \((\eta_G,\eta_H)\) be a \((0,k_2)\)-configuration over \(G\) and \(H\).
    Given a function table for a coloring \(\chi\) on \([[x_{k_1},y_{k_2}]\rightharpoonup V(G)]_{\eta_G}\dotcup[[x_{k_1},y_{k_2}]\rightharpoonup V(H)]_{\eta_H}\),
    a function table for \(\owlref_{(k_1,0)}^{(\infty)}(\chi)\) can be computed in time \(n^{O(k_1)}\) space \(O(k_1n^{k_1}\log n)\).
\end{corollary}

Now, we turn to refinements with respect to non-requantifiable variables.
\begin{lemma}\label{lem:space:owl:non-reusable}
    Let \(k_1+k_2\geq 1\), \(G\) and \(H\) be graphs, and \(\chi\) a coloring on \([[x_{k_1},y_{k_2}]\rightharpoonup V(G)]\dotcup[[x_{k_1},y_{k_2}]\rightharpoonup V(H)]\).

    Given oracle access to the order of \(\chi\)-colors, 
    we can compute for every \((0,k_2)\)-configuration \((\eta_G,\eta_H)\) the function table of \(\owlref_{(0,k_2)}(\chi)\)
    on \([[x_{k_1},y_{k_2}]\rightharpoonup V(G)]_{\eta_G}\dotcup[[x_{k_1},y_{k_2}]\rightharpoonup V(H)]_{\eta_H}\)
    using space \(O(k_1n^{k_1}\log n+ k_2\log n)\) and time \(n^{O(k_1)}\).
\end{lemma}
\begin{proof}
    Note that we have enough space to hold the function table.
    In addition, we will only need logarithmic space.

    Using \Cref{lem:lex:order:multisets}, we can compute the lexicographic ordering
    of \(\owlref_{(0,k_2)}(\chi)\)-colors with logarithmic additional space.
    We can then compute the function table by assigning to each assignment \(\alpha\)
    as the new color
    the number of assignments $\beta$ in
    \([[x_{k_1},y_{k_2}]\rightharpoonup V(G)]_{\eta_G}\dotcup
    [[x_{k_1},y_{k_2}]\rightharpoonup V(H)]_{\eta_H}\)
    such that \[\owlref_{(0,k_2)}(\chi)(\beta)\leq_{\text{lex}}\owlref_{(0,k_2)}(\chi)(\alpha),\]
    which is a number in \([2(n+1)^{k_1}]\).
\end{proof}

Together, these two statements allow us to compare the colors computed by the \((k_1,k_2)\)-dimensional
oblivious Weisfeiler-Leman algorithm in a time- and space-efficient manner.
\begin{theorem}\label{thm:space:complexity}
    Let \(k_1+k_2\geq 1\) be fixed.
    For all \((k_1,k_2)\)-configurations \((\alpha,\beta)\) over graphs \(G\) and \(H\),
    we can decide whether
    \(G,\alpha_1\equiv_{\mathsf{C}^{(k_1,k_2)}} H,\alpha_2\) using space \(O\bigl(k_1(k_2+1)n^{k_1}\log n+(k_2)^2\log n\bigr)\)
    and polynomial time.
\end{theorem}
\begin{proof}
    By \Cref{lem:refinement:nf}, we have
    \[\owl_{(k_1,k_2)}^{(\infty)}(G)\equiv
        \left(\owlref_{(k_1,0)}^{(\infty)}\circ
        \owlref_{(0,k_2)}\right)^{(k_2)}\left(\owl_{(k_1,0)}^{(\infty)}(G)\right).\]
    and similarly for \(H\).
    We show by induction on \(r\) that we can compute for every pair of graphs \(G_1,G_2\in\{G,H\}\)    
    and every \((0,k_2)\)-configuration
    \((\eta_1,\eta_2)\) over \(G_1\) and \(G_2\), a function table of
    \[\chi_r\coloneqq \left(\owlref_{(k_1,0)}^{(\infty)}\circ
        \owlref_{(0,k_2)}\right)^{(r)}\left(\owl_{(k_1,0)}^{(\infty)}\right)\]
    on \([[x_{k_1},y_{k_2}]\rightharpoonup V(G_1)]|_{\eta_1}\dotcup[[x_{k_1},y_{k_2}]\rightharpoonup V(G_2)]|_{\eta_2}\)
    using time \(n^{O((k_1+1)(r+1))}\) and space \(O\bigl(k_1(r+1)n^{k_1}\log n+k_2(r+1)\log n\bigr)\),
    where \(\owl_{(k_1,0)}^{(\infty)}(G_1,G_2)\) is the common coloring computed by \((k_1,k_2)\)-OWL
    on both \(G_1\) and \(G_2\).

    If \(r=0\), we only need to compute the classical OWL-coloring.
    To do this, we first note that we can compute a function table listing the atomic types
    of assignments, each encoded as numbers in \([2(n+1)^{k_1}]\).
    Then, the claim follows from \Cref{cor:space:owl:reusable}.

    For the induction step, assume we can compute function tables
    for (restrictions of) the coloring \(\chi_r\) for every fixed
    \((0,k_2)\)-configuration \((\eta_1,\eta_2)\) over \(G_1\) and \(G_2\)
    in the required time and space.
    This in particular implies that we can compute the order of \(\chi_r\)-colors
    of arbitrary assignments in time \(n^{O((k_1+1)(r+1))}\) and space
    \(O\bigl((r+1)(k_1n^{k_1}\log n+k_2\log n)\bigr)\).

    For every fixed \((0,k_2)\)-configuration \((\eta_1,\eta_2)\),
    we can thus compute a function table for the refined coloring
    \(\owlref_{(0,k_2)}(\chi_r)\) on \([[x_{k_1},y_{k_2}]\rightharpoonup V(G)]_{\eta_1,\eta_2}\)
    using space
    \(O\bigl((r+2)(k_1n^{k_1}\log n+k_2\log n)\bigr)\)
    by \Cref{lem:space:owl:non-reusable}.
    Because the algorithm from \Cref{lem:space:owl:non-reusable} runs in time \(n^{O(k_1)}\),
    it can make at most \(n^{O(k_1)}\) comparisons of previously computed colors, which means that
    this step takes time at most \(n^{O(k_1)}\cdot n^{O((k_1+1)r)}=n^{O((k_1+1)(r+2))}\).
    
    Using \Cref{cor:space:owl:reusable}, we can refine this to a function table of
    \[\chi_{r+1}=\owlref_{(k_1,0)}^{(\infty)}\circ\owlref_{(0,k_2)}(\chi_r).\]

    For \(r=k_2\), this yields the claim.
\end{proof}

\section{Graphs Identified by Logics with Restricted Requantification}\label{sec:planar:and:tree:depth}
One of the reasons for the theoretical interest in finite variable counting logics, in particular for graph isomorphism testing,
is that for many graph classes, some fixed number of variables already suffices
to identify every graph in this graph class. We say that a logic \emph{identifies} a graph $G$ if it distinguishes $G$ from every non-isomorphic graph.
This is e.g.{\@} the case for graphs of bounded rank-width \cite{grohe_rank_width}, graphs of bounded genus \cite{grohe_minors}
and more generally, graphs excluding some fixed minor \cite{grohe_minors,grohe_minors_book}.
In the light of logics with restricted quantification, it is thus natural to ask whether such results can be refined.
That is, we ask whether there are natural graph classes such that already a bounded number of non-requantifiable variables
together with only few requantifiable variables suffice to identify every graph in this class.
We answer this question affirmatively.
First, we show that the logic \(\mathsf{C}^{(0,d+1)}\) identifies all graphs of tree-depth at most \(d\).
Second, we refine previous work in which it was shown that \(\mathsf{C}^{(4,0)}\) identifies all planar graphs \cite{kiefer_planar}.
By closely inspecting the arguments, we show that already \(\mathsf{C}^{(2,2)}\) suffices to identify all \(3\)-connected planar graphs.

\subsection*{Graphs of bounded tree-depth}
Tree-depth is a graph parameter that intuitively measures how close a graph is to a star \cite{nesetril_tree-depth}.
Let $G$ be a graph and let $G_1,\dotsc,G_p$ be the connected components of $G$. Then the \emph{tree-depth} of $G$ is inductively defined as
$$ \td(G) \coloneqq \begin{cases}
        1                                & \text{if } |G| = 1                  \\
        1 + \min_{v \in V(G)} \td(G - v) & \text{if } p=1 \text{ and } |G| > 1 \\
        \max_{i \in [p]} \td(G_i)        & \text{otherwise.}
    \end{cases} $$
Note that a graph has tree-depth \(1\) if and only if it is an independent set, and tree-depth \(2\) if and only if it is a disjoint union of isolated vertices and stars.
Intuitively, graphs of bounded tree-depth are those which, by repeatedly deleting one vertex in each connected component, can be eliminated in a bounded number of rounds. 
We start by showing how to simulate this deletion of vertices by pebbling them instead.

\begin{definition}
    Let $G$ be a graph with vertex coloring $\chi_G$ and let $v \in V(G)$.
    We define the colored graph $G \wr v$ as the graph $G - v$ with the new coloring $\chi_{G \wr v}$ defined by setting
    $$ \chi_{G \wr v}(w) \coloneqq \begin{cases}
            (\chi_G(w), 1) & \text{if } \{v,w\} \in E(G)    \\
            (\chi_G(w), 0) & \text{if } \{v,w\} \notin E(G)
        \end{cases}$$
    for all $w \in V(G) \setminus \{v\}$.
\end{definition}

\begin{lemma}[{\cite[Lemma 4.5]{grohe_tree-depth}}] \label{lem:bp_init}
    Let $k_2 \geq 2$, $G, H$ be graphs with $|G| = |H| \geq 2$ and $v \in V(G), w \in V(H)$ with $\chi_G(v) = \chi_H(w)$.
    Then the following are equivalent:
    \begin{enumerate}
        \item (D) has a winning strategy for $\BP_{(0,k_2)}(G,H)$ with initial position given by
              $\alpha(y_1) = v, \beta (y_1) = w$ and $\dom(\alpha)=\dom(\beta)=\{y_1\}$.
        \item (D) has a winning strategy for $\BP_{(0, k_2-1)}(G \wr v, H \wr w)$.
    \end{enumerate}
\end{lemma}

We can now prove our result on identification of graphs of bounded tree-depth:
\begin{theorem}\label{thm:logic:vs:tree:depth}
    For all $d \geq 1$, the logic $\mathsf{C}^{(0,d+1)}$ identifies all colored graphs of tree-depth at most $d$.
\end{theorem}
\begin{proof}
    For two graphs \(G\) and \(C\), denote by \(\operatorname{noc}(G,C)\) the number of connected components of \(G\) that are isomorphic to \(C\).
    Using the equivalence of the logic and bijective pebble game, it suffices to prove the following claim, which
    lends itself better to an inductive proof:
    \begin{claim}
        Let \(G\) and \(H\) be colored graphs, and \(C\) a connected, colored graph of tree-depth at most \(k_2\).
        If \(\operatorname{noc}(G,C)\neq\operatorname{noc}(H,C)\), then
        (S) has a winning strategy for the game $\BP_{(0,k_2+1)}(G, H)$.
    \end{claim}
    \begin{claimproof}
        We argue by induction on \(k_2\).
        If \(k_2=1\), then \(C\) is a single vertex. Thus, \(G\) and \(H\) differ in the number of isolated vertices of some specific color,
        which allows (S) to win in \(2\) rounds.

        For the induction step, assume that the claim is true for \(k_2\).
        Now, consider a graph \(C\) with \(\operatorname{td}(C)\leq k_2+1\)
        and assume w.l.o.g.~that \(\operatorname{noc}(G,C)>\operatorname{noc}(H,C)\).
        By the definition of tree-depth, \(C\) contains a vertex \(c\) such that \(\operatorname{td}(C-c)\leq k_2\).
        Let \(s\) be the number of such vertices.

        We call a vertex \(v\) in either \(G\) or \(H\) \emph{\(C\)-shrinking} if it is contained
        in a connected component \(C_v\) isomorphic to \(C\), and \(\operatorname{td}(C_v-v)\leq k_2\).
        The number of \(C\)-shrinking vertices in \(G\) and \(H\) is \(s\cdot\operatorname{noc}(G,C)\)
        and \(s\cdot\operatorname{noc}(H,C)\) respectively. In particular, \(G\) contains more \(C\)-shrinking
        vertices than \(H\).

        Now, we describe the winning strategy for (S) in the game $\BP_{(0,k_2+2)}(G, H)$.
        First, (S) picks up the (unused) pebble pair \(y_1\), and (D) picks a bijection \(f\colon V(G)\to V(H)\).
        Then there exist some \(C\)-shrinking vertex \(v\in V(G)\) such that its image \(f(v)\) is not \(C\)-shrinking
        in \(H\). Then (S) places the pebble pair on these two vertices. By \Cref{lem:bp_init}, it now suffices to argue
        that (D) has a winning strategy for the game $\BP_{(0, k_2+1)}(G \wr v, H \wr f(v))$.
        For this, let \(C_v\) be the connected component of \(v\) in \(G\),
        and consider the connected components \(C_v^{(1)},\dots,C_v^{(\ell)}\)of \(C_v\wr v\subseteq G\wr v\).
        If for some \(i\in[\ell]\), we have \(\operatorname{noc}(G \wr v,C_v^{(i)})\neq\operatorname{noc}(H \wr f(v),C_v^{(i)})\), then we are done by
        the induction hypothesis. Thus, we are left with the case that \(\operatorname{noc}(G \wr v,C_v^{(i)})=\operatorname{noc}(H \wr f(v),C_v^{(i)})\).
        Note that the vertex-colorings of \(G\wr v\) and \(H\wr f(v)\) ensure that all connected components
        isomorphic to \(C_v^{(i)}\) for some \(i\) are incident to \(v\) or \(f(v)\) respectively.
        Thus, all copies of \(C_v^{(i)}\) in \(G\) lie in \(C_v\).

        In this case, there is an isomorphism \(\varphi\) between the subgraphs induced by \(G\wr v\) and \(H\wr f(v)\)
        on the union of connected components isomorphic to \(C_v^{(i)}\) for some \(i\in[\ell]\).
        The vertex-colorings of \(G\wr v\) and \(H\wr f(v)\) further ensure that \(\varphi\) can be extended
        by \(\varphi(v)\coloneqq f(v)\), so that its domain is all of \(C_v\). Thus, \(\varphi\) now embeds
        \(C_v\) into the connected component of \(f(v)\), which might, however, have additional vertices
        attached to \(f(v)\). If the image of \(C_v\) under \(\varphi\) was the whole connected component of \(f(v)\),
        then \(f(v)\) would be \(C\)-shrinking, which contradicts our assumption. Thus, there are additional vertices
        attached to \(f(v)\). Thus, \(v\) and \(f(v)\) have distinct degrees, which allows (S) to win in one further round.
    \end{claimproof}
    The theorem now follows by applying the claim to all connected components of \(G\).
\end{proof}
Note that using \Cref{thm:space:complexity} on the space complexity of $\mathsf{C}^{(k_1,k_2)}$-equivalence, the theorem in particular reproves the statement
that isomorphism of graphs of bounded tree-depth can be decided in logarithmic space~\cite{das_logspace}.

Note moreover that because \(\mathsf{C}^{(1,1)}\) can count the number of connected components isomorphic to a fixed colored star by \Cref{lem:deg_seq},
the above inductive proof also shows that for \(d\geq 2\), the logic \(\mathsf{C}^{(1,d-1)}\) identifies all colored graphs of tree-depth at most \(d\).

\subsection*{3-connected planar graphs}
The next class of graphs we consider are \(3\)-connected planar graphs. These naturally appear in the proof that \(\mathsf{C}^4\) identifies
all planar graphs, which starts by reducing the claim for arbitrary planar graphs to \(3\)-connected planar graphs via the decomposition into triconnected components \cite{kiefer_planar}. We recall their proof that \(\mathsf{C}^4\) identifies every \(3\)-connected planar graph
and show that it actually already yields that \(\mathsf{C}^{(2,2)}\) suffices.

The underlying technical lemma is the following:
\begin{lemma}[{\cite[Lemma 23]{kiefer_planar}}] \label{lem:planar_individualize}
    Let $G$ be a $3$-connected planar graph and let $v_1,v_2,v_3 \in V(G)$.
    If $v_1,v_2,v_3$ lie on a common face of $G$, then $\wl^{(\infty)}_{1}(G_{(v_1,v_2,v_3)})$ is a discrete coloring.
\end{lemma}
In the classical setting, from this lemma one can obtain that $4$-WL, i.e., $\mathsf{C}^5$ identifies all $3$-connected planar graphs.
We make use of our framework and show that instead it suffices to use non-reusable resources to cover the individualized vertices.
\begin{corollary}[compare {\cite[Corollary 24]{kiefer_planar}}] \label{cor:identify}
    The logic $\mathsf{C}^{(2,3)}$ identifies all $3$-connected planar graphs.
\end{corollary}
\begin{proof}
    Let \(G\) be a \(3\)-connected planar graph.
    By  \Cref{lem:planar_individualize}, there are vertices $v_1,v_2,v_3 \in V(G)$ such that $\wl^{(\infty)}_{1}(G_{(v_1,v_2,v_3)})$ is a discrete coloring.
    Then by {\Cref{lem:wl:vs:owl}}, also $\owl^{(\infty)}_{(2,0)}(G_{(v_1,v_2,v_3)})$ is a discrete coloring,
    which implies that \(G_{(v_1,v_2,v_3)}\) is identified by \(\mathsf{C}^{(2,0)}\). Thus, there exists a formula \(\varphi(y_1,y_2,y_3)\in \mathsf{C}^{(2,3)}\)
    which defines the graph \(G\) with three individualized vertices represented by \(y_1,y_2\) and \(y_3\) up to isomorphism.
    But then, the formula $\exists y_1\exists y_2\exists y_3 \varphi(y_1,y_2,y_3)$ identifies \(G\).
\end{proof}

In order to improve the identification result from \(\mathsf{C}^5\) to \(\mathsf{C}^4\) and from \(\mathsf{C}^{(2,3)}\) to \(\mathsf{C}^{(2,2)}\),
one observes that for almost all \(3\)-connected graphs, the individualization of just \(2\) vertices already suffices for the above claim,
and the exceptions, where indeed, \(3\) vertices are necessary are somewhat rare.
\begin{definition}
    A \(3\)-connected planar graph is called an \emph{exception} if there are no two vertices \(v_1,v_2\in V(G)\) such that
    \(\wl^{(\infty)}_{1}(G_{(v_1,v_2)})\) is discrete.
\end{definition}
The crucial tool in lowering the number of variables needed is an explicit classification of all exceptions \cite{kiefer_planar}.
This allows to prove the following:
\begin{lemma}\label{lem:exceptions_identified}
    All exceptions are identified by \(\mathsf{C}^{(2,2)}\).
\end{lemma}
\begin{proof}
    The exceptions are the following:
    \begin{itemize}
        \item all bipyramids, i.e., cycles with two additional non-adjacent but otherwise universal vertices,
        \item all platonic solids besides the dodecahedron, and the rhombic dodecahedron,
        \item the triakis tetrahedron, the tetrakis hexahedron and the triakis octahedron, i.e., the graphs obtained
              from the tetrahedron, the hexahedron and the octahedron by adding one vertex per face, whose neighborhood consists of the vertices on that face.
    \end{itemize}
    Because even color refinement identifies every graph with at most \(5\) vertices, and the bipyramid of order \(6\) is the octahedron,
    we start with bipyramids of order at least \(7\). We individualize two adjacent vertices of the cycle. Then, color refinement
    computes a discrete coloring on the underlying cycle, while the two pyramid tips get the same color, which is, however,
    distinct from all other colors. This graph is identified by color refinement.

    Next, consider the platonic solids and the rhombic dodecahedron. All of these are distance-regular graphs of diameter at most \(4\),
    with at most \(14\) vertices. Note that for every \(d\in\mathbb{N}\), there exists a formula
    \(\varphi_d(y_1,y_2)\in\mathsf{C}^{(2,2)}\) stating that \(y_1\) and \(y_2\) have distance
    \(d\). This implies that in \(\mathsf{C}^{(2,2)}\) we can express that a graph \(G\) is distance-regular with a given
    parameter set.
    Because every distance-regular graph of order at most \(14\) is determined by its parameters \cite{distance-regular},
    \(\mathsf{C}^{(2,2)}\) identifies all platonic solids and the rhombic dodecahedron.

    For the last case, note that the vertices of the platonic solid and the added vertices for every face have distinct degrees.
    As moreover, adding these vertices does not change the distance between any two original vertices, \(\mathsf{C}^{(2,2)}\)
    can still express that the underlying platonic solid is of the correct type. Additionally, we can express that the neighborhood
    of each added face vertex is a cycle of the correct length, and that no two face vertices share more than \(2\) common neighbors.
    This identifies the last class of exceptions.
\end{proof}

This finally allows us to prove identification by \(\mathsf{C}^{(2,2)}\):
\begin{theorem}\label{thm:3-con:plan:graph:vs:logic}
    Every \(3\)-connected planar graph is identified by \(\mathsf{C}^{(2,2)}\).
    \begin{proof}
        Let \(G\) be a \(3\)-connected planar graph.
        If \(G\) is an exception, this follows from \Cref{lem:exceptions_identified}.
        If \(G\) is not an exception, the claim can be proven as in \Cref{cor:identify},
        where we instead only need to individualize \(2\) instead of \(3\) vertices.
    \end{proof}
\end{theorem}

It is unclear how precisely this result generalizes to all planar graphs. Moreover,
it is known that all graphs of Euler genus \(g\) are identified by \(\mathsf{C}^{4g+4}\) \cite{grohe_genus},
and we would expect that also here, for sufficiently connected graphs only a very small number of the variables must be requantified.

\section{Outlook}

In this work, we establish a refined framework for the logical description of graphs by means of the requantification of variables. 
We indicate some open questions for future work in vastly different directions. 

Towards structural graph theory, the newly defined cops-and-robber game $\CR^{(k_1,k_2)}$ defines a two-parametric family of graph classes, which contain the tree-width and tree-depth graphs as subclasses. 
It will be interesting to obtain graph-theoretic characterizations for these classes and to study them from an algorithmic and logical point of view. 

From a practical viewpoint, the space complexity is generally the roadblock to a use of higher-dimensional Weisfeiler-Leman in isomorphism testing and graph neural networks.
The introduction of non-requantifiable variables is a technique to limit space complexity, so it needs to be investigated whether problems that arise in practice can be solved by this technique.

In another direction, it will be interesting to investigate other classes of graphs that are known to have bounded Weisfeiler-Leman dimension. 
Using the new variant with restricted reusability, it may be possible to obtain a more fine-grained complexity measure and more space-efficient algorithms for such graph classes.
In particular, it seems highly plausible that (sufficiently connected) bounded genus graphs require only a limited number of requantifiable variables. 
This might allow to design easier fixed-parameter tractable results for graph isomorphism, for example on bounded genus graphs (see \cite{grohe_genus}). 
However, for this there are further restrictions, as non-reusable variables must be choosable from FPT-size bounded sets.

\bibliographystyle{plainurl}
\bibliography{bibliography}

\end{document}